%% file: main.tex
\documentclass[lettersize,onecolumn]{IEEEtran}
\usepackage{amsmath,amsfonts}
\allowdisplaybreaks
\usepackage{array}
\usepackage[caption=false,font=normalsize,labelfont=sf,textfont=sf]{subfig}
\usepackage{textcomp}
\usepackage{stfloats}
\usepackage{url}
\usepackage{verbatim}
\usepackage{graphicx}
\usepackage{cite}
\input{color_package}
\usepackage{pgfplots}
\pgfplotsset{compat=1.18}
\usepackage{bm}
\usepackage{enumitem}
\usepackage{bbm}
\usepackage[ruled, linesnumbered, vlined]{algorithm2e}
\usepackage{booktabs} 
\usetikzlibrary{spy,backgrounds}
\usepackage{placeins}
\usepackage{amssymb}
\usepackage{mathtools}
\usepackage{amsthm}
\usepackage[normalem]{ulem}
\usepackage{float}
\usepackage[acronym]{glossaries}
\usepackage{physics}

\newcounter{mytempeqncnt}

\DeclareMathOperator*{\argmax}{arg\,max}

\newacronym{rm}{RM}{Reed--Muller}
\newacronym{ml}{ML}{maximum-likelihood}
\newacronym{fht}{FHT}{fast Hadamard transform}
\newacronym{rpa}{RPA}{recursive projection-aggregation}
\newacronym{cpa}{CPA}{collapsed projection-aggregation}
\newacronym{pcpa}{PCPA}{pruned collapsed projection-aggregation}
\newacronym{pa}{PA}{projection-aggregation}
\newacronym{scl}{SCL}{successive cancellation list}
\newacronym{bp}{BP}{belief propagation}
\newacronym{ldpc}{LDPC}{low-density parity-check}
\newacronym{llr}{LLR}{log-likelihood ratio}
\newacronym{fer}{FER}{frame error rate}
\newacronym{awgn}{AWGN}{additive white Gaussian noise}
\newacronym{biawgn}{BIAWGN}{binary-input additive white Gaussian noise}
\newacronym{bpsk}{BPSK}{binary phase shift keying}
\newacronym{bsc}{BSC}{binary symmetric channel}
\newacronym{iid}{i.i.d}{independent and identically distributed}
\newacronym{bms}{BMS}{binary-input memoryless symmetric}

\newtheorem{theorem}{Theorem}

\newtheorem{lemma}[theorem]{Lemma}
\newtheorem{prop}{Proposition}
\newtheorem{definition}{Definition}

\newenvironment{customlegend}[1][]{%
    \begingroup
    \csname pgfplots@init@cleared@structures\endcsname
    \pgfplotsset{#1}%
}{%
    \csname pgfplots@createlegend\endcsname
    \endgroup
}%
\def\addlegendimage{\csname pgfplots@addlegendimage\endcsname}

\pgfmathdeclarefunction{gauss}{2}{%
  \pgfmathparse{1/(#2*sqrt(2*pi))*exp(-((x-#1)^2)/(2*#2^2))}%
}

\begin{document}
\title{Density Evolution of Soft-Decision Collapsed Projection-Aggregation Decoding for Reed--Muller Codes over the BIAWGN Channel}

\author{
\IEEEauthorblockN{Jiajie Li, Marvin Rübenacke, and Warren J. Gross}\\
\thanks{
Jiajie Li, Marvin Rübenacke, and Warren J. Gross are with the Department of Electrical and Computer Engineering, McGill University, Montr{\'e}al, Qu{\'e}bec, Canada.
(e-mail: jiajie.li@mail.mcgill.ca; marvin.ruebenacke@mcgill.ca; warren.gross@mcgill.ca).
}
}




\maketitle

\begin{abstract}
\gls{rm} codes have been shown to achieve capacity over a range of channels, and recently proposed \gls{pa} decoding has been experimentally shown to achieve near-maximum-likelihood decoding performance.
These recent achievements motivate theoretical research on \gls{pa} decoding.
In this work, we analyze the density function of the soft output from \gls{cpa} decoding for \gls{rm} codes over the \gls{biawgn} channel.
We prove that soft-decision \gls{cpa} decoding returns an exact marginal probability and is symmetric. 
Based on the analysis, we build a density evolution model for \gls{cpa} decoding.
To simplify the density evolution, we approximate the projection and the fast Hadamard transform decoding using hard-decision decoding.
Simulation results over the \gls{biawgn} channel show that our proposed density evolution model captures the fast reduction in the mean and the variance of the soft information returned from the CPA decoding, which qualitatively explains the decoding mechanism and the fast convergence speed of the CPA decoding.
We perform an asymptotic analysis based on the proposed density evolution, and we show that \gls{cpa} decoding can achieve a vanishing error probability for \gls{rm} codes with a vanishing code rate.
\end{abstract}


\glsresetall 

\section{Introduction}
Recent research demonstrates that \gls{rm} codes~\cite{RMcode} achieve channel capacity under a wide range of channels, such as the binary erasure channel~\cite{RMBEC}, the \gls{bsc}~\cite{RMBSC, RMBSCimprove}, and the \gls{bms} channel~\cite{RMBMS_bit_journal,RMBMS_block,abbe2024polynomial,pfister2025capacity}.
In addition to the capacity-achieving capability, the \gls{rm} codes are structurally similar to polar codes~\cite{rmpolar} and also exhibit a polarization effect~\cite{rmpolarize}.

Decoding with affordable complexity for \gls{rm} codes is the key to utilizing the desired characteristic mentioned.
The first decoding algorithm for \gls{rm} codes is majority-vote decoding~\cite{reed} that can correct error patterns with a weight of fewer than half of the minimum distance of \gls{rm} codes.
For order $r=1$ \gls{rm} codes, \gls{ml} decoding performance can be achieved under a complexity $O\left(n\log_{2}\left(n\right)\right)$ using the \gls{fht} decoding~\cite{FHT,FHT_soft}, where $n$ is the code length.
Many decoding algorithms are proposed for \gls{rm} codes with $r\geq2$.
For example, Dumer's recursive list decoding can achieve \gls{ml} decoding performance given a sufficiently large list size~\cite{DumerList}.

The recently proposed \gls{rpa} decoding and its list decoding are observed to achieve near-\gls{ml} decoding performance for a range of code lengths and code rates~\cite{RPA}.
Given its near-\gls{ml} decoding performance, theoretical analysis on \gls{rpa} decoding is conducted in the literature~\cite{rameshwar2024upper,fathollahi2026error,zhang2025errorpatternpa}.
It is proven in~\cite{rameshwar2024upper} that the \gls{rpa} decoding can asymptotically achieve vanishing probability over the \gls{bsc} for \gls{rm} codes with $r\leq\log\left(cm\right)$, where $m$ is the code length parameter and $c$ is a constant that is proportional to the cross-over probability of the \gls{bsc}.
Later, the result is extended to the \gls{bms} channel~\cite{fathollahi2026error}.

A variant of \gls{rpa} decoding, namely \gls{cpa} decoding, is proposed to reduce the computational complexity of \gls{rpa} decoding by removing repeated subspaces~\cite{RPA_BP}.
\gls{cpa} decoding targets the soft-decision variant of \gls{rpa} decoding and, as pointed out in~\cite{RPA_BP}, it is closely related to \gls{bp} decoding.
The similarity of the decoding results under different subspaces in \gls{cpa} decoding is first analyzed in~\cite{Pruning}.
The error patterns encountered during \gls{rpa} and \gls{cpa} decoding are analyzed in~\cite{zhang2025errorpatternpa}, where it is further shown that \gls{rpa} and \gls{cpa} decoding with $2t+1$ subspaces can correct up to $t$ errors.
Moreover, a recent work shows that the \gls{pa} decoding (\gls{rpa} decoding and its variants) can decode error patterns with a weight of fewer than half the minimum distance efficiently~\cite{zhang2025minimum}.
Also, \gls{cpa} decoding is further simplified in~\cite{li2024layeredCPA}, where the extrinsic inter-iteration update is replaced by the broadcast inter-iteration update.

Theoretical research~\cite{rameshwar2024upper,fathollahi2026error,zhang2025errorpatternpa,zhang2025minimum} on \gls{rpa} and \gls{cpa} decoding focuses on hard-decision decoding, where the number of errors encountered by decoding algorithms is investigated, and analysis tools for soft-decision decoding, which handle the probability density function of received soft information, are still missing in the literature.
In this work, we analyze soft-decision \gls{cpa} decoding, propose a density evolution model to track changes in the density function of soft-decision outputs from \gls{cpa} decoding, and perform an asymptotic analysis using the proposed density evolution model.
The following contributions on soft-decision \gls{cpa} decoding are made:
\begin{enumerate}
    \item We found that \gls{cpa} decoding returns the exact marginal probability. Also, we prove that \gls{cpa} decoding is symmetric with respect to the transmitted codeword.
    \item We propose density evolution to analyze the density function of the soft-decision output from each iteration of \gls{cpa} decoding over the \gls{biawgn} channel.
    Proposed density evolution captures the fast reduction in the mean and variance of the soft information returned from \gls{cpa} decoding, which qualitatively explains the decoding mechanism and the fast convergence of \gls{cpa} decoding.
    \item Based on the proposed density evolution, we perform an asymptotic analysis and find that the soft-decision \gls{cpa} decoding achieves vanishing error probability when the \gls{rm} code has a vanishing code rate.
\end{enumerate}

This work is structured as follows.
Section~\ref{sec:preliminaries} shows the necessary background of \gls{cpa} decoding.
Section~\ref{sec:keyprop_CPA} presents key properties that facilitate the condition for applying the density evolution.
Section~\ref{sec:DE_CPA} presents the density evolution model for the~\gls{cpa} decoding.
Section~\ref{sec:Asym_Ana_CPA} is the asymptotic analysis of \gls{cpa} decoding based on our proposed density evolution model.
Section~\ref{sec:conclusion} concludes this work.

\section{Preliminaries}
\label{sec:preliminaries}
\subsection{Notations}
Matrices and vectors are denoted as bold upper-case letters ($\bm{M}$) and bold lower-case letters ($\bm{v}$), respectively. 
The transpose operation is $^\top$, and a projection based on the coset of a subspace $\mathbbm{B}_{i}$ is denoted by the subscript $/\mathbbm{B}_{i}$. 
Binary indices are denoted by the letter $z$.
The probability of an event is denoted by $\mathbbm{P}\left(\cdot\right)$.
The probability density function is denoted by $\mathrm{p}\left(\cdot\right)$.
\subsection{RM Codes}
\gls{rm}$\left(m,r\right)$ codes are $\left(n,k\right)$ linear codes, where $n=2^{m}$ is the code length, $k=\sum_{i=0}^{i=r}\binom{m}{i}$ is the code dimension, $0 \leq r \leq m$, and the code rate equals $R=\tfrac{k}{n}$.
The generator matrix for encoding the \gls{rm}$\left(m,r\right)$ codes can be constructed in the following two steps:
\begin{enumerate}
    \item The generator matrix for the RM$(m,m)$ code ($\bm{G}(m,m)$) is obtained by applying the $m$-th Kronecker power of the base matrix $\bm{F}$ \cite{rmpolar}:
    \begin{equation}
        \bm{G}(m,m)=\bm{F}^{\otimes m}, \text{ }\bm{F}=
        \begin{bmatrix}
        1&0\\
        1&1
        \end{bmatrix}\text{.}
    \end{equation}
    \item Select rows with the $k$ largest Hamming weights, which are at least $2^{m-r}$, in $\bm{G}\left(m,m\right)$ to compose the generator matrix $\bm{G}\left(m,r\right)$ for the \gls{rm}$\left(m,r\right)$ code.
\end{enumerate}
The dual code of the \gls{rm}$\left(m,r\right)$ code is the \gls{rm}$\left(m,m-r-1\right)$ code, so the parity-check matrix $\bm{H}$ of the RM$(m,r)$ code is $\bm{H}=\bm{G}(m,m-r-1)$~\cite{RM_S}.
Valid codewords $\bm{c}\in \text{RM}(m,r)$ produce an all-zeros syndrome vector $\bm{s}=\bm{H}\bm{c}^{\top}=\bm{0}$.

\subsection{Collapsed Projection-Aggregation Decoding}
\gls{cpa} decoding is a three-step iterative decoding algorithm, namely:

\begin{enumerate}
    \item Projection: In this step, the received \gls{llr} vector $\bm{l}$ of the \gls{rm}$\left(m,r\right)$ code is used to compute the \gls{llr} of the \gls{rm}$\left(m-r+1,1\right)$ codes by projecting to the $\left(r-1\right)$-dimensional subspace $\mathbbm{B}_{i}$
    \begin{equation}
        \begin{split}
            \bm{l}_{/\mathbbm{B}_{i}}(T)
            &=2\tanh^{-1}\left(\underset{z\in T}{\prod}\tanh\left(\frac{\bm{l}(z)}{2}\right)\right)\text{,}\\
        \end{split}
    \label{eqn:cpallrproj}
    \end{equation}
    where $T$ is the coset of the subspace $\mathbbm{B}_{i}$ of dimensions $r-1$.
    It is shown in~\cite[Thm. 2]{Li2023improvePAlist} that the $\left(r-1\right)$-dimensional subspaces can be decomposed into $r-1$ $1$-dimensional subspaces:
    \begin{equation}
    \begin{split}
        &h_{2^{r-1}}(l_{1},l_{2},...,l_{2^{r-1}})\\ &= 2\tanh^{-1}\left(\prod_{i=1}^{2^{r-1}}\tanh\left(\frac{l_{i}}{2}\right)\right)\\
        &= h_{2} (h_{2}(h_{2}...), h_{2}(h_{2}...))\text{,}
    \end{split}
    \label{eqn:dimension_decomposition}
    \end{equation}
    where $h_{2}\left(A,B\right) = 2\tanh^{-1}\left(\tanh\left(\frac{A}{2}\right)\tanh\left(\frac{B}{2}\right)\right)$ is the projection function for the $1$-dimensional subspace, which is also known as the \emph{box-plus} ($\boxplus$) operator~\cite{hagenauer2002iterative}, and $l_{i}\in\bm{l}$.
    There are $n_{\mathbbm{B}}=\binom{m}{r-1}_{2}$ different $(r-1)$-dimensional subspaces in \gls{cpa} decoding.
    \item \gls{fht} decoding: The projected \gls{llr} vector of \gls{rm}$\left(m-r+1,1\right)$ codes is decoded by \gls{ml} \gls{fht} decoding, and the decoded codeword is defined as 
    \begin{equation}
        \hat{\bm{y}}_{/\mathbbm{B}_{i}} = \argmax_{\bm{c}\in\mathrm{\gls{rm}}\left(m-r+1,1\right)}\mathbbm{P}\left(\bm{l}_{/\mathbbm{B}_{i}} \mid \bm{c}\right)\text{,}
        \label{eqn:def_ml_fht}
    \end{equation}
    where $\bm{l}_{/\mathbbm{B}_{i}}$ is the \gls{llr} vector after the projection function.
    The decoded bit $\hat{\bm{y}}_{/\mathbbm{B}_{i}}\!\left(T\right)$ estimates the parity check of received code bits in the same coset $T$.
    \item Aggregation: Given the received \gls{llr} vector and the decoding result from the \gls{fht} decoding, new \gls{llr}s of the received code bits are computed by
    \begin{equation}
    \begin{split}
        &\hat{\bm{l}}\left(z\right)\\
        &=\sum_{i=1}^{n_{B}} (-1)^{\hat{\bm{y}}_{/\mathbbm{B}_{i}}(T)}\\
        &\left(2\tanh^{-1}\left(\underset{{z}_{j}\in T \setminus \{z\}}{\prod}\tanh\left(\frac{\bm{l}(z_{j})}{2}\right)\right)\right)\\
        &\stackrel{(a)}{=}\sum_{i=1}^{n_{B}} 2\tanh^{-1}\left( \tanh\left(\frac{+\infty \times (-1)^{\hat{\bm{y}}_{/\mathbbm{B}_{i}}(T)}}{2}\right)\right.\\
        &\left.\underset{{z}_{j}\in T \setminus \{z\}}{\prod}\tanh\left(\frac{\bm{l}({z}_{j})}{2}\right)\right)\text{,}
    \end{split}
    \label{eqn:llraggr}
    \end{equation}
    where the equality $(a)$ in~\eqref{eqn:llraggr} holds according to~\cite[equation $(28)$]{Li2023improvePAlist}.
    The average of the aggregated \gls{llr} vector ($\bar{\bm{l}}=\hat{\bm{l}}/ n_{\mathbbm{B}}$) is either fed to the next iteration, or hard decisions of $\bar{\bm{l}}$ is returned as the decoded codeword if the maximum number ($N_{\text{max}}$) of iterations or an early-stopping criterion (e.g., the difference in the L2 norm between soft information returned from two consecutive iterations is smaller than a threshold $\theta$~\cite{RPA,RPA_BP}) is met.
\end{enumerate}

\section{Revisiting \gls{cpa} Decoding}
\label{sec:keyprop_CPA}
Before building the density evolution model, several key properties should be verified so that this analysis tool can be used.
For example, the density evolution for \gls{ldpc} codes is performed under the assumption of an infinite code length~\cite{richardson2001deldpc} because factor graphs of \gls{ldpc} codes will have a tree structure~\cite[Sec. 3.7.2]{richardson2008modern}, and the exact marginal probability can be computed by the \gls{bp} decoding when the factor graph structure is a tree~\cite[Sec. 2.2]{richardson2008modern}.

We first show that the three-step process of \gls{cpa} decoding resembles the exact marginal probability, and results returned from different subspaces can be viewed as different realizations of the random variable.
Also, to simplify the analysis process, we usually assume the all-zeros or all-ones codeword is sent.
Hence, we need to show that \gls{cpa} decoding is symmetric regardless of the transmitted codeword.
While the symmetry of the \gls{rpa} decoding is proved in~\cite{RPA}, the proof of the symmetry of \gls{cpa} decoding is missing in the literature, and we prove the symmetry property of \gls{cpa} decoding by modifying the proof for the \gls{rpa} decoding in~\cite{RPA}.

\subsection{Exact Marginal Probability Computed By \gls{cpa} Decoding}
In this section, we would like to show that \gls{cpa} decoding estimates the exact marginal probability, so we do not need to find whether there are cycles present in the Tanner graph structure for \gls{cpa} decoding~\cite{RPA_BP}.

The \gls{llr} $\bm{l}\left(z\right)$ used in the first iteration is defined as 
\begin{equation}
    \bm{l}\left( z \right) = \ln\left( \frac{\mathbbm{P}\left( \mathsf{Y}\left(z\right) = \bm{y}\left(z\right) \mid  \bm{c}\left(z\right) = 0 \right) }{\mathbbm{P}\left( \mathsf{Y}\left(z\right) = \bm{y}\left(z\right) \mid  \bm{c}\left(z\right) = 1 \right)} \right)\text{,}
    \label{eqn:def_llr}
\end{equation}
where $\mathsf{Y}\left(z\right)$ is the output random variable at index $z$ from the \gls{biawgn} channel, and $\bm{y}\left(z\right)$ is the realization of $\mathsf{Y}\left(z\right)$.
The following proposition shows the meaning behind the projection function.
\begin{prop}
    The projection to the $\left(r-1\right)$-dimensional subspace is to compute 
    \begin{equation}
        \ln\left( \frac{\mathbbm{P}\left( \left\{\mathsf{Y}\left(z\right)=\bm{y}\left(z\right)\mid z\in T\right\} \mid  \oplus_{z\in T} \bm{c}\left(z\right) = 0 \right)}{\mathbbm{P}\left( \left\{\mathsf{Y}\left(z\right)=\bm{y}\left(z\right)\mid z\in T\right\} \mid  \oplus_{z\in T} \bm{c}\left(z\right) = 1 \right)} \right)\text{,}
    \end{equation}
    where $T$ is a coset of the $\left(r-1\right)$-dimensional subspace $\mathbbm{B}_{i}$, and $z$ is the index in the coset $T$.
    \label{prop:proj_llr}
\end{prop}
\begin{proof}
    Let $z_{1}$ and $z_{2}$ denote two different indices. It is shown in~\cite{RPA} that the projection to the one-dimensional subspace can be computed by
    \begin{equation}
        \begin{split}&\bm{l}_{/1}\!\left(z_{1},z_{2}\right)\\
        &=\ln\left(  \frac{\exp\left(\bm{l}\left(z_{1}\right)+\bm{l}\left( z_{2} \right)\right)+1}{\exp\left(\bm{l}\left(z_{1}\right)\right) +\exp\left(\bm{l}\left(z_{2}\right)\right)}\right)\\
        &\stackrel{(a)}{=} 2\tanh^{-1}\left(\tanh\left(\frac{\bm{l}\left(z_{1}\right)}{2}\right)\tanh\left(\frac{\bm{l}\left(z_{2}\right)}{2}\right)\right)\\
        &\stackrel{(b)}{=}\ln\left( \frac{\mathbbm{P}\left(\mathsf{Y}_{1}=y_{1},\mathsf{Y}_{2}=y_{2}\mid  \bm{c}\left(z_{1}\right)\oplus\bm{c}\left(z_{2}\right)=0\right)  }{\mathbbm{P}\left(\mathsf{Y}_{1}=y_{1},\mathsf{Y}_{2}=y_{2}\mid  \bm{c}\left(z_{1}\right)\oplus\bm{c}\left(z_{2}\right)=1\right)} \right)\text{,}
        \end{split}
        \label{eqn:proj_1d_physical_meaning}
    \end{equation}
    where $y_{i}$ is the received symbol corresponding to $\bm{l}\left(z_{i}\right)$, and $\mathsf{Y}_{i}$ is the random variable.
    The equality $(a)$ in~\eqref{eqn:proj_1d_physical_meaning} is shown in~\cite{RPA_BP}, and the equality $(b)$ in~\eqref{eqn:proj_1d_physical_meaning} is shown in~\cite[equation $\left(12\right)$]{RPA}.
    
    It is also shown in~\eqref{eqn:dimension_decomposition} that the projection of the $\left(r-1\right)$-dimensional subspace can be decomposed into the composition of the projection on the $1$-dimensional subspace.
    Hence, by induction, the projection to the $2$-dimensional subspace can be decomposed into the composition of the projection to the $1$-dimensional subspace.
    Let $z_{1}$ and $z_{2}$ denote the indices used in the first $1$-dimensional projection, and $z_{3}$ and $z_{4}$ denote indices used in the second $1$-dimensional projection.
    Let $A$ and $B$ denote the events $\left(\mathsf{Y}_{1}=y_{1},\mathsf{Y}_{2}=y_{2}\right)$ and $\left(\mathsf{Y}_{3}=y_{3},\mathsf{Y}_{4}=y_{4}\right)$, respectively, and $\bm{c}_{/1}\!\left(z_{i},z_{j}\right) := \bm{c}\left(z_{i}\right)\oplus\bm{c}\left(z_{j}\right)$.
    The $2$-dimensional projection involved $4$ different \gls{llr}s can be computed by applying the $1$-dimensional projection on the results returned from the previous two $1$-dimensional projections
    \begin{equation}
        \begin{split}&2\tanh^{-1}\left(\tanh\left(\frac{\bm{l}_{/1}\!\left(z_{1},z_{2}\right)}{2}\right)\tanh\left(\frac{\bm{l}_{/1}\!\left(z_{3},z_{4}\right)}{2}\right)\right)\\
        &\stackrel{(a)}{=}\ln\left( \frac{\mathbbm{P}\left(A,B\mid \bm{c}_{/1}\!\left(z_{1},z_{2}\right)\oplus\bm{c}_{/1}\!\left(z_{3},z_{4}\right)=0\right)  }{\mathbbm{P}\left(A,B\mid \bm{c}_{/1}\!\left(z_{1},z_{2}\right)\oplus\bm{c}_{/1}\!\left(z_{3},z_{4}\right)=1\right)} \right)\\
        &=\ln\left( \frac{\mathbbm{P}\left(\mathsf{Y}_{1}=y_{1},...,\mathsf{Y}_{4}=y_{4}\mid \oplus_{i=1}^{4}\bm{c}\left(z_{i}\right)=0\right)  }{\mathbbm{P}\left(\mathsf{Y}_{1}=y_{1},...,\mathsf{Y}_{4}=y_{4}\mid \oplus_{i=1}^{4}\bm{c}\left(z_{i}\right)=1\right)} \right)\text{,}
        \end{split}\label{eqn:proj_2d_physical_meaning}
    \end{equation}
    where the equality $(a)$ in~\eqref{eqn:proj_2d_physical_meaning} is hold by the definition of the projection on the $1$-dimensional subspace~\eqref{eqn:proj_1d_physical_meaning}.

    By induction, the projection on the $\left(r-1\right)$-dimensional subspace is to compute
    \begin{equation}
    \begin{split}
    &\bm{l}_{/\mathbbm{B}_{i}}(T)\\
    &=2\tanh^{-1}\left(\underset{z\in T}{\prod}\tanh\left(\frac{\bm{l}(z)}{2}\right)\right)\\
    &=\ln\left( \frac{\mathbbm{P}\left( \left\{\mathsf{Y} \left(z\right)=\bm{y}\left(z\right)\mid z\in T\right\} \mid  \oplus_{z\in T} \bm{c}\left(z\right) = 0 \right)}{\mathbbm{P}\left( \left\{\mathsf{Y}\left(z\right)=\bm{y}\left(z\right)\mid  z\in T\right\} \mid  \oplus_{z\in T} \bm{c}\left(z\right) = 1 \right)} \right)\text{.}
    \end{split}
    \end{equation}
\end{proof}
The meaning of this $\tanh\left(\cdot\right)$ update rule is also well-investigated in the literature, such as in~\cite{hagenauer2002iterative} (i.e., box-plus $\boxplus$), and used in theoretical analysis of decoding algorithms~\cite{richardson2001deldpc,saeyoung2001Gaussianapproximation}.
More details regarding this $\tanh\left(\cdot\right)$ update rule can refer to~\cite{saeyoung2001Gaussianapproximation,luby1998analysis,forney2002codes,hartmann1976optimum,battail1979replication,hagenauer2002iterative}.
To fit into the context of the \gls{pa} decoding, we extend the proof in~\cite{RPA} to \gls{cpa} decoding to support the analysis that will be done in this work.

As shown in~\eqref{eqn:llraggr}, the output $\hat{\bm{y}}_{/\mathbbm{B}_{i}}\!\left(T\right)$ from the \gls{fht} decoding can be viewed as having a \gls{llr}
\begin{equation}
\begin{split}
    &\ln \left( \frac{\mathbbm{P}\left(\bm{l}_{/\mathbbm{B}_{i}}\mid \hat{\bm{y}}_{/\mathbbm{B}_{i}},\hat{\bm{y}}_{/\mathbbm{B}_{i}}\!\left(T\right)=0\right)}{\mathbbm{P}\left(\bm{l}_{/\mathbbm{B}_{i}}\mid \hat{\bm{y}}_{/\mathbbm{B}_{i}},\hat{\bm{y}}_{/\mathbbm{B}_{i}}\!\left(T\right)=1\right)} \right)\\
    &=\ln \left( \frac{\mathbbm{P}\left(\bm{l}_{/\mathbbm{B}_{i}},\hat{\bm{y}}_{/\mathbbm{B}_{i}}\mid \hat{\bm{y}}_{/\mathbbm{B}_{i}}\!\left(T\right)=0\right)/\mathbbm{P}\left(\hat{\bm{y}}_{/\mathbbm{B}_{i}}\right)}{\mathbbm{P}\left(\bm{l}_{/\mathbbm{B}_{i}},\hat{\bm{y}}_{/\mathbbm{B}_{i}}\mid \hat{\bm{y}}_{/\mathbbm{B}_{i}}\!\left(T\right)=1\right)/\mathbbm{P}\left(\hat{\bm{y}}_{/\mathbbm{B}_{i}}\right)} \right)\\
    &\stackrel{(a)}{=}\ln \left( \frac{\mathbbm{P}\left(\bm{l}_{/\mathbbm{B}_{i}}\mid \hat{\bm{y}}_{/\mathbbm{B}_{i}}\!\left(T\right)=0\right)}{\mathbbm{P}\left(\bm{l}_{/\mathbbm{B}_{i}}\mid \hat{\bm{y}}_{/\mathbbm{B}_{i}}\!\left(T\right)=1\right)} \right)\\
    &\stackrel{(b)}{=}\ln\left( \frac{\mathbbm{P}\left( \mathsf{Y} \left(z\right) = \bm{y}\left(z\right)\;\forall z\in\{0,1\}^{m} \mid \hat{\bm{y}}_{/\mathbbm{B}_{i}}\!\left(T\right) = 0 \right) }{\mathbbm{P}\left( \mathsf{Y} \left(z\right) = \bm{y}\left(z\right)\;\forall z\in\{0,1\}^{m} \mid \hat{\bm{y}}_{/\mathbbm{B}_{i}}\!\left(T\right) = 1 \right)} \right)\\
    &\in \left\{ +\infty,-\infty \right\}\text{,}
\end{split}
\label{eqn:llr_def_fht}
\end{equation}
where the equality $(a)$ holds because the value of $\hat{\bm{y}}_{/\mathbbm{B}_{i}}$ is the output from the \gls{fht} decoding with a probability of $1$ given the input $\bm{y}_{/\mathbbm{B}_{i}}$ and a fixed ordering in the code book, and the interception of the event $A$ of $\bm{l}_{/\mathbbm{B}_{i}}$ and the event $B$ of $\hat{\bm{y}}_{/\mathbbm{B}_{i}}$ is $A\bigcap B= A$, and the equality $(b)$ holds because the $\bm{l}_{/\mathbbm{B}_{i}}$ and the $\hat{\bm{y}}_{/\mathbbm{B}_{i}}$ are the fixed returned solution given a fixed ordering of the code book and a $\mathbbm{B}_{i}$.

Based on the derivation of~\eqref{eqn:llr_def_fht}, we reformulate the aggregation function~\eqref{eqn:llraggr} in~\eqref{eqn:aggr_marginalprob}.
Then, by Proposition~\ref{prop:proj_llr} and the meaning of the results from \gls{fht} decoding~\eqref{eqn:llr_def_fht}, the aggregation function is to compute the equality $\left(a\right)$ in~\eqref{eqn:aggr_marginalprob}.
\begin{figure*}
[!t]
\normalsize
\setcounter{mytempeqncnt}{\value{equation}}
\setcounter{equation}{12}
\begin{equation}
\begin{split}
    &2\tanh^{-1}\left( \tanh\left(\frac{+\infty \times (-1)^{\hat{\bm{y}}_{/\mathbbm{B}_{i}}(T)}}{2}\right)\underset{{z}_{i}\in T \setminus \{z\}}{\prod}\tanh\left(\frac{\bm{l}({z}_{i})}{2}\right)\right)\\
    &\stackrel{(a)}{=}\ln\left( \frac{\mathbbm{P}\left( \left\{\mathsf{Y} \left(z'\right)=\bm{y}\left(z'\right)\mid z'\in T\setminus \{z\}\right\}, \mathsf{Y} \left(z\right) = \bm{y}\left(z\right)\forall z\in\{0,1\}^{m}\mid  \left(\oplus_{z'\in T\setminus \{z\}} \bm{c}\left(z'\right) \right) \oplus \hat{\bm{y}}_{/\mathbbm{B}_{i}}\!\left(T\right) = 0 \right)}{\mathbbm{P}\left( \left\{\mathsf{Y} \left(z'\right)=\bm{y}\left(z'\right)\mid z'\in T\setminus \{z\}\right\}, \mathsf{Y}\left(z\right) = \bm{y}\left(z\right)\forall z\in\{0,1\}^{m}\mid  \left(\oplus_{z'\in T\setminus \{z\}} \bm{c}\left(z'\right) \right) \oplus \hat{\bm{y}}_{/\mathbbm{B}_{i}}\!\left(T\right) = 1 \right)}\right)\\
    &=\ln\left( \frac{\mathbbm{P}\left( \left\{\mathsf{Y} \left(z'\right)=\bm{y}\left(z'\right)\mid z'\in T\setminus \{z\}\right\}, \mathsf{Y} \left(z\right) = \bm{y}\left(z\right)\forall z\in\{0,1\}^{m}\mid \bm{c}\left(z\right)= 0 \right)}{\mathbbm{P}\left( \left\{\mathsf{Y} \left(z'\right)=\bm{y}\left(z'\right)\mid z'\in T\setminus \{z\}\right\}, \mathsf{Y}\left(z\right) = \bm{y}\left(z\right)\forall z\in\{0,1\}^{m}\mid \bm{c}\left(z\right) = 1 \right)}\right)\\
    &\stackrel{(b)}{=}\ln\left(  \frac{\left(\prod_{z'\in \left\{0,1\right\}^{m}\setminus \{z\} } \mathbbm{P}\left(\mathsf{Y}\left(z'\right)=\bm{y}\left(z'\right)\right)\right) \mathbbm{P}\left(\mathsf{Y}\left(z\right)=\bm{y}\left(z\right)\mid \bm{c}(z)=0\right)}{\left(\prod_{z'\in \left\{0,1\right\}^{m}\setminus \{z\} } \mathbbm{P}\left(\mathsf{Y}\left(z'\right)=\bm{y}\left(z'\right)\right)\right) \mathbbm{P}\left(\mathsf{Y}\left(z\right)=\bm{y}\left(z\right)\mid \bm{c}(z)=1\right)} \right)=\ln\left( \frac{\mathbbm{P}\left(\mathsf{Y}\left(z\right)=\bm{y}\left(z\right)\mid \bm{c}\left(z\right) = 0\right)}{\mathbbm{P}\left(\mathsf{Y}\left(z\right)=\bm{y}\left(z\right)\mid  \bm{c}\left(z\right) = 1\right)} \right)\text{.}\\
    \end{split}
    \label{eqn:aggr_marginalprob}
\end{equation}
\setcounter{equation}{13}
\hrulefill
\vspace*{4pt}
\end{figure*}
The equality $(b)$ in~\eqref{eqn:aggr_marginalprob} holds because $\mathsf{Y}(z)$ is independent of $\mathsf{Y}(z')$, and $\mathsf{Y}\left(z'\right)$ is independent of $\bm{c}\left(z\right)$ for $z\neq z'$.
From~\eqref{eqn:aggr_marginalprob}, we can see that the aggregation recovers the exact marginal probability, and the summation and the averaging step aim to find an accurate estimation of $\bm{l}\left(z\right)$ across all automorphisms.

To mimic the \gls{bp} decoding, an extrinsic update
\begin{equation}
\begin{split}
    \bm{l}_{\rightarrow\mathbbm{B}_{i}}(z) &= \frac{1}{n_{\mathbbm{B}}} 
    \sum_{j \in \{ 1,2,...,n_{\mathbbm{B}}\} \setminus \{i\}}    
    (-1)^{\hat{\bm{y}}_{/\mathbbm{B}_{j}}(T)}\\
        &\left(2\tanh^{-1}\left(\underset{{z}_{k}\in T \setminus \{z\}}{\prod}\tanh\left(\frac{\bm{l}(z_{k})}{2}\right)\right)\right)\text{,}
\end{split}
\end{equation}
where $\bm{l}_{\rightarrow\mathbbm{B}_{i}}$ is the \gls{llr} vector used by the subspace $\mathbbm{B}_{i}$ in the second iteration,
is used to update the \gls{llr} for the next decoding iteration for \gls{cpa} decoding~\cite{RPA}. 
The \gls{bp} decoding excludes the variable information from itself to avoid repeated counting when computing the marginal information.
However, by~\eqref{eqn:aggr_marginalprob}, a realization of the random variable in the marginal probability is recovered in each subspace; hence, there is no need to remove the information decoded by the subspace $\mathbbm{B}_{i}$ when updating the \gls{llr} for the next iteration.
The broadcast update proposed in~\cite{li2024layeredCPA}
\begin{equation}
    \bm{l}_{\rightarrow\mathbbm{B}_{i}}(z) = \frac{1}{n_{\mathbbm{B}}}\hat{\bm{l}}\left(z\right)
    \label{eqn:broadcastllrupdate}
\end{equation}
uses the same \gls{llr} vector for all subspaces for iterations larger than one.
Combining with the marginal probability derived in~\eqref{eqn:aggr_marginalprob}, we can safely use the broadcast update~\eqref{eqn:broadcastllrupdate} and assume that inputs for all subspaces are the same for iterations larger than one.
\subsection{Symmetry of \gls{cpa} Decoding}
To simplify the theoretical analysis, the all-zero codeword can be assumed to be transmitted if the decoding is symmetric (i.e., decoding will have the same decoding performance regardless of the transmitted codeword).
We prove that \gls{cpa} decoding is symmetric by Proposition~\ref{prop:symmetric_CPA}, which implies that we can use the all-zeros codeword assumption in this work.
Similar to the proof in~\cite{RPA}, the following \gls{ml} metric
\begin{equation}
    \quad\frac{1}{2}\sum_{z\in\mathbbm{E}:=\mathbbm{F}_{2}^{m}}\left((-1)^{\bm{c}(z)}\bm{l}\left(z\right)\right)\text{,}
    \label{eqn:likelihood_list}
\end{equation}
which is the same as~\cite[Eq. (9)]{RPA}, is defined and is used in the proof.
\begin{lemma}
    Let $\bm{c}_{0}=\left( \bm{c}_{0}\left(z\right)\text{, } z\in\mathbbm{E}\right)$, where $\mathbbm{E}:=\mathbbm{F}_{2}^{m}$, be a codeword of the $\text{RM}\left(m,r\right)$ code.
    Let $\bm{l}^{(1)}=\left( \bm{l}^{(1)}(z)\text{, } z\in\mathbbm{E}\right)$ and $\bm{l}^{(2)}=\left( \bm{l}^{(2)}(z)\text{, } z\in\mathbbm{E}\right)$ be two LLR vectors such that 
    \begin{equation}
        \quad\bm{l}^{(2)}(z) = (-1)^{\bm{c}_{0}\left(z\right)} \bm{l}^{(1)}(z)\quad\forall z\in\mathbbm{E}\text{.}
        \label{eqn:sym_LLR}
    \end{equation}
    Denote $\hat{\bm{c}}_{1}=\mathrm{CPA}\left( \bm{l}^{(1)}, m, r, N_{\text{max}}\right)$ and $\hat{\bm{c}}_{2}=\mathrm{CPA}\left( \bm{l}^{(2)}, m, r, N_{\text{max}}\right)$.
    Then $\hat{\bm{c}}_{2} = \hat{\bm{c}}_{1} + \bm{c}_{0}$.
    \label{lemma:symmetric_decoded_llr}
\end{lemma}
\begin{proof}
    The proof starts with the base case $r=1$ using the ML decoding, the \gls{fht} decoding, and the procedure is similar to the proof in~\cite{RPA}.
    Hence, according to~\eqref{eqn:likelihood_list}, $\hat{\bm{c}}_{2}=\text{CPA}\left( \bm{l}^{(2)}, m, 1, N_{\text{max}}\right)$ is the codeword in $\text{RM}(m,1)$, and 
    \begin{equation}
    \begin{split}
    &\sum_{z\in\mathbbm{E}}\left((-1)^{\hat{\bm{c}}_{2}(z)}\bm{l}^{\left(2\right)}\left(z\right)\right) \geq \sum_{z\in\mathbbm{E}}\left((-1)^{\bm{c}(z)}\bm{l}^{\left(2\right)}\left(z\right)\right)\\
    &\forall \bm{c}\in \text{RM}\left(m,1\right)\text{.}
    \end{split}
    \end{equation}
    By~\eqref{eqn:sym_LLR}, $\forall \bm{c}\in \text{RM}\left(m,1\right)$, we have
    \begin{equation}
    \begin{split}
        \sum_{z\in\mathbbm{E}}\left((-1)^{\hat{\bm{c}}_{2}(z)\oplus\bm{c}_{0}\left(z\right)}\bm{l}^{\left(1\right)}\left(z\right)\right) \geq \sum_{z\in\mathbbm{E}}\left((-1)^{\bm{c}(z)\oplus\bm{c}_{0}\left(z\right)}\bm{l}^{\left(1\right)}\left(z\right)\right)\text{.}
    \end{split}
    \end{equation}
    Because $\bm{c}_{0}\in \text{RM}\left(m,1\right)$, $\bm{c}_{0}+ \text{RM}\left(m,1\right)=\text{RM}\left(m,1\right)$.
    Hence, $\forall\bm{c}\in\text{RM}\left(m,1\right)$,
    \begin{equation}
    \begin{split}
        \sum_{z\in\mathbbm{E}}\left((-1)^{\hat{\bm{c}}_{2}(z)\oplus\bm{c}_{0}\left(z\right)}\bm{l}^{\left(1\right)}\left(z\right)\right) \geq \sum_{z\in\mathbbm{E}}\left((-1)^{\bm{c}\left(z\right)}\bm{l}^{\left(1\right)}\left(z\right)\right)\text{.}
    \end{split}
    \end{equation}
    Therefore, for the base case $r=1$, we can conclude that $\hat{\bm{c}}_{1} = \hat{\bm{c}}_{2} + \bm{c}_{0}$ is the codeword in $\text{RM}\left(m,1\right)$ that maximize
    \begin{equation}
        \sum_{z\in\mathbbm{E}}\left((-1)^{\bm{c}\left(z\right)}\bm{l}^{\left(1\right)}\left(z\right)\right)\text{.}
    \end{equation}

    For $r>1$, the decoded codeword $\hat{\bm{c}}\left(z\right)$ is determined by the sign of $\bar{\bm{l}}\left(z\right)$.
    If we can show that the updated \gls{llr} vectors $\bar{\bm{l}}^{\left(1\right)}$ and $\bar{\bm{l}}^{\left(2\right)}$ satisfy~\eqref{eqn:sym_LLR}, then $\hat{\bm{c}}_{1} = \hat{\bm{c}}_{2} + \bm{c}_{0}$.

    We know that
    \begin{equation}
    \begin{split}
        &\bar{\bm{l}}^{\left(j\right)}(z)= \frac{1}{n_{\mathbbm{B}}}\sum_{i=1}^{n_{\mathbbm{B}}}\alpha_{j}\left(\mathbbm{B}_{i},T\right)\\
        &\left(2\tanh^{-1}\left(\underset{z_{k}\in T \setminus \{z\}}{\prod}\tanh\left(\frac{\bm{l}^{\left( j\right)}(z_{k})}{2}\right)\right)\right)\text{,}
    \end{split}
    \end{equation}
    where $\alpha_{j}\left(\mathbbm{B}_{i},T\right)=(-1)^{\hat{\bm{y}}_{/\mathbbm{B}_{i}}(T)}$ for $z\in T$, and $T$ is the coset under the subspace $\mathbbm{B}_{i}$.
    Then we show that $\alpha_{2}\left(\mathbbm{B}_{i},T\right) = (-1)^{\oplus_{z\in T} \bm{c}_{0}\left( z\right)}\alpha_{1}\left(\mathbbm{B}_{i},T\right)$.
    It can be seen that $\alpha_{j}\left(\mathbbm{B}_{i}, T\right)$ is determined by $\hat{\bm{y}}_{/\mathbbm{B}_{i}} = \text{CPA}\left(\bm{l}_{/\mathbbm{B}_{i}}^{\left(j\right)}, m-r+1, 1, N_{\text{max}}, \theta\right)$.
    
    From the projection function~\eqref{eqn:cpallrproj}, we can see that
    \begin{equation}
    \begin{split}
    &\bm{l}_{/\mathbbm{B}_{i}}^{\left(2\right)}\left(T\right)\\
    &=2\tanh^{-1}\left(\underset{z\in T}{\prod}\tanh\left(\frac{\bm{l}^{\left(2\right)}(z)}{2}\right)\right)\\
    & = 2\tanh^{-1}\left(\underset{z\in T}{\prod}\tanh\left(\frac{(-1)^{\bm{c}_{0}\left(z\right)} \bm{l}^{(1)}(z)}{2}\right)\right)\\
    & = (-1)^{\oplus_{z\in T} \bm{c}_{0}\left( z\right)} 2\tanh^{-1}\left(\underset{z\in T}{\prod}\tanh\left(\frac{ \bm{l}^{(1)}(z)}{2}\right)\right)\\
    & = (-1)^{\oplus_{z\in T} \bm{c}_{0}\left( z\right)} \bm{l}_{/\mathbbm{B}_{i}}^{\left(1\right)}(T)\text{.}
    \end{split}
    \end{equation}
    Let $\bm{c}_{0}\left(T\right) := \oplus_{z\in T} \bm{c}_{0}\left( z\right)$, we have $\left( \bm{c}_{0}\left(T\right), T\in \mathbbm{E}/ \mathbbm{B}\right)\in \text{RM}\left(m-r+1,1\right)$.
    Hence, the codeword $\left( \bm{c}_{0}\left(T\right), T\in \mathbbm{E}/ \mathbbm{B}_{i}\right)$ and two \gls{llr} vectors $\left( \bm{l}_{/\mathbbm{B}_{i}}^{\left(1\right)}(T), T\in \mathbbm{E}/\mathbbm{B}_{i}\right)$ and  $\left( \bm{l}_{/\mathbbm{B}_{i}}^{\left(2\right)}(T), T\in \mathbbm{E}/\mathbbm{B}_{i}\right)$ satisfy the condition of the base case of this lemma.
    It can be concluded $\alpha_{2}\left(\mathbbm{B}_{i},T\right) = (-1)^{\oplus_{z\in T} \bm{c}_{0}\left( z\right)}\alpha_{1}\left(\mathbbm{B}_{i},T\right)$, and 
    \begin{equation}
    \begin{split}
        &\bar{\bm{l}}^{\left(2\right)}\left(z\right)\\
        &= \frac{1}{n_{\mathbbm{B}}}\\
        &\sum_{i=1}^{n_{\mathbbm{B}}}\alpha_{2}\left(\mathbbm{B}_{i},T\right)\left(2\tanh^{-1}\left(\underset{z_{j}\in T \setminus \{z\}}{\prod}\tanh\left(\frac{\bm{l}^{\left( 2\right)}(z_{j})}{2}\right)\right)\right)\\
        &=\frac{1}{n_{\mathbbm{B}}}\sum_{i=1}^{n_{\mathbbm{B}}}\left((-1)^{\oplus_{z_{j}\in T} \bm{c}_{0}\left( z_{j}\right)}\right)\alpha_{1}\left(\mathbbm{B}_{i},T\right)\left((-1)^{\oplus_{z_{j}\in T\setminus \{z\}}\bm{c}_{0}\left(z_{j}\right)}\right)\\
        &\left(2\tanh^{-1}\left(\underset{z_{j}\in T \setminus \{z\}}{\prod}\tanh\left(\frac{\bm{l}^{\left( 1\right)}(z_{j})}{2}\right)\right)\right)\\
        &=\left(-1\right)^{\bm{c}_{0}\left(z\right)}\frac{1}{n_{\mathbbm{B}}}\\
        &\sum_{i=1}^{n_{\mathbbm{B}}}\alpha_{1}\left(\mathbbm{B}_{i},T\right) \left(2\tanh^{-1}\left(\underset{z_{j}\in T \setminus \{z\}}{\prod}\tanh\left(\frac{\bm{l}^{\left( 1\right)}(z_{j})}{2}\right)\right)\right)\\
        &=\left(-1\right)^{\bm{c}_{0}\left(z\right)}\bar{\bm{l}}^{\left(1\right)}\left(z\right)\text{.}
    \end{split}
    \end{equation}
\end{proof}
\begin{definition}
    A memoryless channel $W:\left\{0,1 \right\}\rightarrow\mathcal{W}$ is a \gls{bms} channel if there is a permutation $\pi$ of the output alphabet $\mathcal{W}$ such that $\pi^{-1}=\pi$ and $W\left(x\mid 1\right)=W\left(\pi\left(x\right)\mid 0\right)\text{ }\forall x\in \mathcal{W}$~\cite{RPA}. 
    \label{def:BMS_channel}
\end{definition}
\begin{prop}
    Let $W:\left\{0,1\right\}\rightarrow\mathcal{W}$ be a \gls{bms} channel. 
    Let $\bm{c}_{1}$ and $\bm{c}_{2}$ be two codewords of the $\text{RM}\left(m,r\right)$ code.
    Let $\mathsf{Y}_{1}$ and $\mathsf{Y}_{2}$ be the (random) channel outputs of transmitting $\bm{c}_{1}$ and $\bm{c}_{2}$ over $n=2^{m}$ independent copies of $W$, respectively.
    Let $\bm{l}^{\left(1\right)}$ and $\bm{l}^{\left(2\right)}$ be the \gls{llr} vectors corresponding to $\mathsf{Y}_{1}$ and $\mathsf{Y}_{2}$, respectively. 
    Then, for any $\bm{c}_{1},\bm{c}_{2}\in\text{RM}\left(m,r\right)$, we have
    \begin{equation}
    \begin{split}
        &\mathbbm{P}\left(\mathrm{\gls{cpa}}\left(\bm{l}^{\left(1\right)}, m, r, N_{\text{max}}, \theta\right)\neq \bm{c}_{1}\right)\\
        &= \mathbbm{P}\left(\mathrm{\gls{cpa}}\left(\bm{l}^{\left(2\right)}, m, r, N_{\text{max}}, \theta\right)\neq \bm{c}_{2}\right)\text{.}
    \end{split}
    \end{equation}
    \label{prop:symmetric_CPA}
\end{prop}
\begin{proof}
    In this proof, we apply an identical strategy as in~\cite{RPA}.
    Because $W$ is a \gls{bms} channel, there exists a permutation $\pi$ that satisfies two conditions in the Definition~\ref{def:BMS_channel}.
    Let $\bm{c}_{1}$ and $\bm{c}_{2}$ be two codeword of $\text{RM}\left(m,r\right)$ codes, and $\bm{c}_{0}=\bm{c}_{1}+\bm{c}_{2}$ is also a codeword of $\text{RM}\left(m,r\right)$ codes.
    Both channel outputs $\bm{y}_{1}$ and $\bm{y}_{2}$ belong to $\mathcal{W}^{n}$.
    Define a permutation $\pi^{\bm{c}_{0}}$ on $\mathcal{W}^{n}$: For any $\bm{y}=\left(\bm{y}\left( z\right), z\in\mathbbm{E}\right)\in\mathcal{W}^{n}$,
    \begin{equation}
        \pi^{\bm{c}_{0}}\left( \bm{y}\right) = \left(\pi^{\bm{c}_{0}}\left( \bm{y}\left(z\right)\right),z\in\mathbbm{E} \right)\text{.}
    \end{equation}
    The code bit $\bm{c}_{0}\left(z\right)\in \left\{0,1 \right\}$, and $\pi^{\bm{c}_{0}}$ is the identity map.
    Because $\pi$ is a permutation on $\mathcal{W}^{n}$, $\pi^{\bm{c}_{0}}$ is clearly a permutation on $\mathcal{W}^{n}$.

    For a given $\bm{y}=\left(\bm{y}\left(z\right),z\in\mathbbm{E}\right)$, the corresponding \gls{llr} vector is $\bm{l}_{y}^{\left(1 \right)}:=\left(\bm{l}_{y}^{\left(1 \right)}\left(z\right),z\in\mathbbm{E}\right)$ (i.e., $\bm{l}_{y}^{\left(1 \right)}\left(z\right)=\text{LLR}\left(\bm{y}\left(z\right)\right)\;\forall z\in\mathbbm{E}$), and the corresponding LLR vector of $\pi^{\bm{c}_{0}}\left( \bm{y}\right)$ is $\bm{l}_{y}^{\left(2 \right)}:=\left(\bm{l}_{y}^{\left(2 \right)}\left(z\right),z\in\mathbbm{E}\right)$ (i.e., $\bm{l}_{y}^{\left(2 \right)}\left(z\right)=\text{LLR}\left(\pi^{\bm{c}_{0}\left(z\right)}\left(\bm{y}\left(z\right)\right)\right)\;\forall z\in\mathbbm{E}$).
    By the property of Definition~\ref{def:BMS_channel}, we have
    \begin{equation}
        \bm{l}_{y}^{\left(2\right)} \left( z \right) = \left( -1\right)^{\bm{c}_{0}\left(z\right)} \bm{l}_{y}^{\left(1\right)}\!\left( z\right)\quad\forall z\in\mathbbm{E}\text{.}
    \end{equation}

    Since $\bm{c}_{0}\in\text{RM}\left(m,r\right)$, by Lemma~\ref{lemma:symmetric_decoded_llr}, we have
    \begin{equation}
        \text{\gls{cpa}}\!\left( \bm{l}^{(1)}_{y}, m, r, N_{\text{max}},\theta\right) = \text{\gls{cpa}}\!\left( \bm{l}^{(2)}_{y}, m, r, N_{\text{max}}, \theta\right) + \bm{c}_{0}\text{.}
    \end{equation}
    Hence, $\text{\gls{cpa}}\!\left( \bm{l}^{(1)}_{y}, m, r, N_{\text{max}}, \theta\right)\neq\bm{c}_{1}$ if and only if $\text{\gls{cpa}}\!\left( \bm{l}^{(2)}_{y}, m, r, N_{\text{max}}, \theta\right)\neq\bm{c}_{2}$.

    We use $W^{n}\left( \bm{y}\mid \bm{c} \right)$ to denote the probability of receiving $\bm{y}\in\mathcal{W}^{n}$ when the transmitted codeword is $\bm{c}$.
    By the property of $\pi$, we can see that
    \begin{equation}
        W^{n}\left( \bm{y}\mid \bm{c}_{1} \right) = W^{n}\left( \pi^{\bm{c}_{0}}\left(\bm{y}\right)\mid \bm{c}_{2} \right) \quad \forall \bm{y}\in\mathcal{W}^{n}\text{.}
    \end{equation}

    Vectors $\bm{l}^{\left(1\right)}$ and $\bm{l}^{\left(2\right)}$ denote the random \gls{llr} vectors corresponding to random channel outputs when transmitting $\bm{c}_{1}$ and $\bm{c}_{2}$, respectively.
    Then,
    \begin{equation}
    \begin{split}
        &\mathbbm{P}\left(\text{\gls{cpa}}\!\left(\bm{l}^{\left(1\right)}, m, r, N_{\text{max}}, \theta\right)\neq \bm{c}_{1}\right)\\
        &=\sum_{\bm{y}\in\mathcal{W}^{n}}W^{n}\left( \bm{y}\mid \bm{c}_{1}\right)\mathbbm{1}\!\left(  \text{\gls{cpa}}\!\left( \bm{l}^{(1)}_{y}, m, r, N_{\text{max}},\theta\right)\neq\bm{c}_{1}\right)\\
        &=\sum_{\bm{y}\in\mathcal{W}^{n}}W^{n}\left( \pi^{\bm{c}_{0}}\left(\bm{y}\right)\mid \bm{c}_{2}\right)\mathbbm{1}\!\left(  \text{\gls{cpa}}\!\left( \bm{l}^{(2)}_{y}, m, r, N_{\text{max}},\theta\right)\neq\bm{c}_{2}\right)\\
        &=\mathbbm{P}\left(\text{\gls{cpa}}\!\left(\bm{l}^{\left(2\right)}, m, r, N_{\text{max}}, \theta\right)\neq \bm{c}_{2}\right)\text{.}
    \end{split}
    \end{equation}
\end{proof}

\section{Density Evolution of \gls{cpa} Decoding}
\label{sec:DE_CPA}
In this section, we first analyze the density function returned from the projection function and the \gls{fht} decoding for \gls{cpa} decoding.
As the analytical performance of the soft-decision \gls{fht} decoding under the channel condition returned from the projection function is hard to derive, we first approximate the \gls{awgn} channel as a \gls{bsc} to simplify the analysis and use a hard-decision decoding to approximate the \gls{fht} decoding.
Then, we analyze the density function returned from the aggregation.
Lastly, the density evolution is constructed.

\subsection{Approximation in the Projection and the \gls{fht} Decoding}
Because the exact \gls{ml} decoding (i.e., \gls{fht} decoding) performance of the $\text{RM}\left(m-r+1,1\right)$ sub-codes is hard to analytically derive, we take the following simplification in this work.
We assume a \gls{bpsk} modulation, an \gls{awgn} channel and that our \gls{fht} decoding uses the hard-decision input.
We use the \gls{bsc} to approximate the \gls{awgn} channel by setting the crossover probability of the \gls{bsc} to
\begin{equation}
    p=\int_{-\infty}^{0}\mathrm{p}\left(\bm{l}\left(z\right)\mid \bm{c}\left(z\right)=0\right) \dd{\bm{l}\left(z\right)}\text{.}
    \label{eqn:bit_error_prob}
\end{equation}

Based on the projection function~\eqref{eqn:cpallrproj}, the projected \gls{llr} is negative if there is an odd number of received \glspl{llr} $\bm{l}<0$.
Hence, the probability of a projected \gls{llr} $<0$ is
\begin{equation}
\begin{split}
    &\sum_{i\in\mathcal{I}} \binom{2^{r-1}}{i} p^{i}\left(1-p\right)^{2^{r-1}-i}\\
    &=\frac{1}{2}\left(\left(\sum_{i=0}^{2^{r-1}}\binom{2^{r-1}}{i} p^{i}\left(1-p\right)^{2^{r-1}-i}\right)-\right.\\
    &\left.\left(\sum_{i=0}^{2^{r-1}}(-1)^{i}\binom{2^{r-1}}{i} p^{i}\left(1-p\right)^{2^{r-1}-i}\right)\right)\\
    &\stackrel{(a)}{=}\frac{1}{2}\left(\left(\left(1-p\right)+p\right)^{2^{r-1}} - \left(\left(1-p\right)-p\right)^{2^{r-1}}\right)\\
    &=\frac{1}{2}\left(1-\left(1-2p\right)^{2^{r-1}}\right)=\Bar{p}\text{,}
\end{split}
    \label{eqn:BSC_approx_proj_prob}
\end{equation}
where $\mathcal{I}:=\left\{1,3,...,2^{r-1}-1\right\}$ is the set of odd numbers, and the equality $(a)$ holds because of the binomial theorem~\cite{roy2021series}.

The upper bounds of the error probability of the \gls{fht} decoding under the \gls{bsc} and the \gls{bms} channel are derived in~\cite{rameshwar2024upper,fathollahi2026error}, but the tightness of these bounds is not provided, and we do not use these bounds in this work.
For the \gls{bsc}, an efficient decoding algorithm for \gls{rm} codes with an order $r=o\left(\sqrt{m}\right):=\lim_{m\rightarrow\infty}\tfrac{r}{\sqrt{m}}=0$ is proposed to correct error patterns with a weight up to $\left(\frac{1}{2}-o\left(1\right)\right)n$~\cite{saptharishi2016efficiently,saptharishi2017efficiently}.
For $\text{RM}\left(m-r+1,1\right)$ codes that can be decoded by \gls{cpa} decoding ($2\leq r<m$), 
\begin{align*}
    \lim_{m\rightarrow\infty}\tfrac{1}{\sqrt{m-r+1}}=\lim_{m\rightarrow\infty}\tfrac{1}{\sqrt{m\left(1-r/m\right)+1}}=0
\end{align*}
because $0<\left(1-r/m\right)<1$.
Hence, $\frac{n}{2}-1$ errors for order-$r=1$ \gls{rm} codes can be corrected because $1=\tfrac{1}{2^{m}}n$ and $\lim_{m\rightarrow\infty}\tfrac{1}{2^{m}}=0=o\left(1\right)$.
Hence, the probability of the correct decoding is
\begin{equation}
    \mathbbm{P}\left(\hat{\bm{y}}_{/\mathbbm{B}_{i}}=\bm{0}\right) = \sum_{i=0}^{2^{m-r}-1}\binom{2^{m-r}}{i}\Bar{p}^{i}\left(1-\Bar{p}\right)^{2^{m-r+1}-i}\text{.}
    \label{eqn:upper_bound_error_bsc}
\end{equation}
As this efficient decoding algorithm will not have better decoding performance than the \gls{ml} decoding (\gls{fht} decoding), we can safely use it to approximate the decoding performance (an upper bound that will be explained in the next subsection) of the \gls{fht} decoding in the analysis instead of the actual decoding.

\subsection{The Density Function Returned from the Aggregation}
In this work, we propose the following mathematical model for the aggregation function.
We define the output from the aggregation function as
\begin{equation}
\begin{split}
    u&:=(-1)^{\hat{\bm{y}}_{/\mathbbm{B}_{i}}\!\left(T\right)} \times 2\tanh^{-1}\left(\underset{{z}_{i}\in T \setminus \{z\}}{\prod}\tanh\left(\frac{\bm{l}(z_{i})}{2}\right)\right)\\
    &=(-1)^{\hat{\bm{y}}_{/\mathbbm{B}_{i}}\!\left(T\right)} \times l\text{,}
\end{split}
\label{eqn:aggr_var_def}
\end{equation}
and we assume the random variable, which is a \gls{llr}, 
\begin{equation}
    l:=2\tanh^{-1}\left(\underset{{z}_{i}\in T \setminus \{z\}}{\prod}\tanh\left(\frac{\bm{l}(z_{i})}{2}\right)\right)
    \label{eqn:def_aggr_without_fht}
\end{equation}
follows a density function $\mathrm{p}\left(l\right)$.
Because the random variable $\hat{\bm{y}}_{/\mathbbm{B}_{i}}\!\left(T\right)$ and the random variable $l$ are not necessarily independent, the distribution of the product variable $u$ is defined by the joint distribution where $\mathrm{p}\left(u\right):=\mathrm{p}\left((-1)^{\hat{\bm{y}}_{/\mathbbm{B}_{i}}\!\left(T\right)},l\right)$ and 
\begin{equation}
\begin{split}
    &\mathrm{p}\left(u\mid (-1)^{\hat{\bm{y}}_{/\mathbbm{B}_{i}}\!\left(T\right)}\right)\\
    &=\mathrm{p}\left((-1)^{\hat{\bm{y}}_{/\mathbbm{B}_{i}}\!\left(T\right)},l\mid (-1)^{\hat{\bm{y}}_{/\mathbbm{B}_{i}}\!\left(T\right)}\right)\\
    &=\mathrm{p}\left(l\mid (-1)^{\hat{\bm{y}}_{/\mathbbm{B}_{i}}\!\left(T\right)}\right)\text{.}
\end{split}
\end{equation}
Also, the random variable $(-1)^{\hat{\bm{y}}_{/\mathbbm{B}_{i}}\!\left(T\right)}$ only changes the sign of $l$, so we have $\mathrm{p}\left(l\right) = \mathrm{p}\left(l\mid \hat{\bm{y}}_{/\mathbbm{B}_{i}}\!\left(T\right)=0\right)$, and $\mathrm{p}\left(-l\right) = \mathrm{p}\left(l\mid \hat{\bm{y}}_{/\mathbbm{B}_{i}}\!\left(T\right)=1\right)$.
Hence, we have the following:
\begin{equation}
\begin{split}
    \mathrm{p}\left(u\right) =\begin{cases}
        \mathrm{p}\left(l\right)\mathbbm{P}\left(\hat{\bm{y}}_{/\mathbbm{B}_{i}}\!\left(T\right)=0\right)\text{,}&\text{ if }\hat{\bm{y}}_{/\mathbbm{B}_{i}}\!\left(T\right)=0\text{,}\\
        \mathrm{p}\left(-l\right)\mathbbm{P}\left(\hat{\bm{y}}_{/\mathbbm{B}_{i}}\!\left(T\right)=1\right)\text{,}&\text{ if }\hat{\bm{y}}_{/\mathbbm{B}_{i}}\!\left(T\right)=1\text{.}\\
    \end{cases}
    \end{split}
    \label{eqn:product_of_vars_aggreagtions_prob}
\end{equation}

The expected value of $u$ is defined as
\begin{equation}
\begin{split}
    &\mathbbm{E}\left[u\right]\\
    &=  \int_{-\infty}^{+\infty}\sum_{\hat{\bm{y}}_{/\mathbbm{B}_{i}}\!\left(T\right)=0}^{1}\mathrm{p}\left(u\right)\left(-1\right)^{\hat{\bm{y}}_{/\mathbbm{B}_{i}}\!\left(T\right)}l\dd{l}\\
    &= \int_{-\infty}^{+\infty}\sum_{\hat{\bm{y}}_{/\mathbbm{B}_{i}}\!\left(T\right)=0}^{1}\mathrm{p}\left((-1)^{\hat{\bm{y}}_{/\mathbbm{B}_{i}}\!\left(T\right)},l\right)\left(-1\right)^{\hat{\bm{y}}_{/\mathbbm{B}_{i}}\!\left(T\right)}l\dd{l}\\
    &=\int_{-\infty}^{+\infty} \mathrm{p}\left(l\right)\mathbbm{P}\left(\hat{\bm{y}}_{/\mathbbm{B}_{i}}\!\left(T\right)=0\right) l \dd{l} +\\
    &\int_{-\infty}^{+\infty} \mathrm{p}\left(-l\right)\mathbbm{P}\left(\hat{\bm{y}}_{/\mathbbm{B}_{i}}\!\left(T\right)=1\right) \left(-l\right) \dd{l}\\
    &= \mathbbm{P}\left(\hat{\bm{y}}_{/\mathbbm{B}_{i}}\!\left(T\right)=0\right) \mathbbm{E}\left[l\right] + \mathbbm{P}\left(\hat{\bm{y}}_{/\mathbbm{B}_{i}}\!\left(T\right)=1\right) \mathbbm{E}\left[-l\right]\text{.}
\end{split}
\label{eqn:expected_value_u}
\end{equation}
The variance of $u$ is defined as
\begin{equation}
    \text{Var}\left[u\right] = \mathbbm{E}\left[u^{2}\right] - \left(\mathbbm{E}\left[u\right]\right)^{2}\text{,}
\end{equation}
\begin{equation}
\begin{split}
    &\mathbbm{E}\left[u^{2}\right]\\
    &= \int_{-\infty}^{+\infty} \mathrm{p}\left(l\right)\mathbbm{P}\left(\hat{\bm{y}}_{/\mathbbm{B}_{i}}\!\left(T\right)=0\right) \left(l\right)^{2} \dd{l} +\\
    &\int_{-\infty}^{+\infty} \mathrm{p}\left(-l\right)\mathbbm{P}\left(\hat{\bm{y}}_{/\mathbbm{B}_{i}}\!\left(T\right)=1\right) \left(-l\right)^{2} \dd{l}\\ 
    &=\mathbbm{P}\left(\hat{\bm{y}}_{/\mathbbm{B}_{i}}\!\left(T\right)=0\right)\mathbbm{E}\left[l^{2}\right] + \mathbbm{P}\left(\hat{\bm{y}}_{/\mathbbm{B}_{i}}\!\left(T\right)=1\right)\\
    &\left(\int_{-\infty}^{0}\mathrm{p}\left(-l\right)\left(-l\right)^{2}\dd l + \int_{0}^{\infty}\mathrm{p}\left(-l\right)\left(-l\right)^{2}\dd l\right)\\
    &=\mathbbm{P}\left(\hat{\bm{y}}_{/\mathbbm{B}_{i}}\!\left(T\right)=0\right)\mathbbm{E}\left[l^{2}\right] + \mathbbm{P}\left(\hat{\bm{y}}_{/\mathbbm{B}_{i}}\!\left(T\right)=1\right)\\
    &\left(\int_{0}^{\infty}\mathrm{p}\left(l\right)l^{2}\dd l + \int_{-\infty}^{0}\mathrm{p}\left(l\right)l^{2}\dd l\right)\\
    & = \mathbbm{P}\left(\hat{\bm{y}}_{/\mathbbm{B}_{i}}\!\left(T\right)=0\right)\mathbbm{E}\left[l^{2}\right] \\
    & + \mathbbm{P}\left(\hat{\bm{y}}_{/\mathbbm{B}_{i}}\!\left(T\right)=1\right) \mathbbm{E}\left[l^{2}\right]= \mathbbm{E}\left[l^{2}\right]\text{,}
    \end{split}
\end{equation}
and
\begin{equation}
\begin{split}
    &\left(\mathbbm{E}\left[u\right]\right)^{2}\\
    &= \left(\mathbbm{P}\left(\hat{\bm{y}}_{/\mathbbm{B}_{i}}\!\left(T\right)=0\right) \mathbbm{E}\left[l\right]\right)^{2}\\
    &+ \left(\mathbbm{P}\left(\hat{\bm{y}}_{/\mathbbm{B}_{i}}\!\left(T\right)=1\right) \mathbbm{E}\left[-l\right]\right)^{2}\\
    & + 2\mathbbm{P}\left(\hat{\bm{y}}_{/\mathbbm{B}_{i}}\!\left(T\right)=0\right) \mathbbm{E}\left[l\right]\mathbbm{P}\left(\hat{\bm{y}}_{/\mathbbm{B}_{i}}\!\left(T\right)=1\right) \mathbbm{E}\left[-l\right]\text{.}
\end{split}
\end{equation}
Hence, we have
\begin{equation}
\begin{split}
    &\text{Var}\left[u\right]\\
    &= \mathbbm{E}\left[l^{2}\right] + \left(2\mathbbm{P}\left(\hat{\bm{y}}_{/\mathbbm{B}_{i}}\!\left(T\right)=0\right) \mathbbm{P}\left(\hat{\bm{y}}_{/\mathbbm{B}_{i}}\!\left(T\right)=1\right) \right.\\
    &\left.-\mathbbm{P}^{2}\left(\hat{\bm{y}}_{/\mathbbm{B}_{i}}\!\left(T\right)=0\right)-\mathbbm{P}^{2}\left(\hat{\bm{y}}_{/\mathbbm{B}_{i}}\!\left(T\right)=1\right)\right)\left(\mathbbm{E}\left[l\right]\right)^{2}\text{.}
\end{split}
\end{equation}
By~\eqref{eqn:def_ml_fht}, we know that $\hat{\bm{y}}_{/\mathbbm{B}_{i}}\!\left(T\right)$ maximize the likelihood function, and we define success rate of the \gls{ml} decoding as $\mathbbm{P}\left(\hat{\bm{y}}_{/\mathbbm{B}_{i}}=\bm{0}=\bm{c}_{0}\right)$ where $\bm{0}$ is an all-zeros vector and we denote the all-zeros codeword as $\bm{c}_{0}$.

Given all codewords are \gls{iid}, we define
\begin{equation}
    \begin{split}
        &\mathbbm{P}\left(\hat{\bm{y}}_{/\mathbbm{B}_{i}}\!\left(T\right)=0\right)\\
        &= \mathbbm{P}\left(\bigcup_{\hat{\bm{y}}\in\mathcal{C}_{0}}\hat{\bm{y}}\right)\\
        &=\sum_{\hat{\bm{y}}\in\mathcal{C}_{0}} \mathbbm{P}\left(\hat{\bm{y}}\right)\geq \mathbbm{P}\left(\hat{\bm{y}}_{/\mathbbm{B}_{i}}=\bm{0}\right)\text{,}
    \end{split}
\end{equation}
where $\mathcal{C}_{0}:=\left\{\hat{\bm{y}}\in\text{\gls{rm}}\left(m-r+1,1\right)\mid \hat{\bm{y}}\left(T\right)=0\right\}$,
and
\begin{equation}
    \begin{split}
        &\mathbbm{P}\left(\hat{\bm{y}}_{/\mathbbm{B}_{i}}\!\left(T\right)=1\right)\\ 
        &= \mathbbm{P}\left(\bigcup_{\hat{\bm{y}}\in\mathcal{C}_{1}}\hat{\bm{y}}\right)\\
        &=\sum_{\hat{\bm{y}}\in\mathcal{C}_{1}} \mathbbm{P}\left(\hat{\bm{y}}\right)\leq\mathbbm{P}\left(\hat{\bm{y}}_{/\mathbbm{B}_{i}}\neq \bm{0}\right)\text{,}
    \end{split}
\end{equation}
where $\mathcal{C}_{1}:=\left\{\hat{\bm{y}}\in\text{\gls{rm}}\left(m-r+1,1\right)\mid \hat{\bm{y}}\left(T\right)=1\right\}$.
By only considering the event that the \gls{ml} decoding successfully decodes (i.e., $\hat{\bm{y}}_{/\mathbbm{B}_{i}}=\bm{0}$), and we can see that
\begin{equation}
\begin{split}
    \mathbbm{E}\left[u\right]&\geq \mathbbm{P}\left(\hat{\bm{y}}_{/\mathbbm{B}_{i}}=\bm{0}\right) \mathbbm{E}\left[l\right] + \mathbbm{P}\left(\hat{\bm{y}}_{/\mathbbm{B}_{i}}\neq \bm{0}\right) \mathbbm{E}\left[-l\right]\\
    &=\mu'\text{.}
\end{split}
    \label{eqn:sum_llr_mean}
\end{equation}
As the decoding performance of the efficient hard-decision decoding is worse than the decoding performance of the \gls{fht} decoding, the mean $\mu'$ will be larger when \gls{fht} decoding is used, and the error rate of using hard-decision decoding~\cite{saptharishi2016efficiently,saptharishi2017efficiently} is an upper bound on the error rate of using \gls{fht} decoding.

Also, an upper bound of $\text{Var}\left[u\right]$ can be constructed by
\begin{equation}
    \begin{split}
        &\text{Var}\left[u\right]\\
        &= \mathbbm{E}\left[l^{2}\right] + \left(2\mathbbm{P}\left(\hat{\bm{y}}_{/\mathbbm{B}_{i}}\!\left(T\right)=0\right) \mathbbm{P}\left(\hat{\bm{y}}_{/\mathbbm{B}_{i}}\!\left(T\right)=1\right) \right.\\
        &\left.-\mathbbm{P}^{2}\left(\hat{\bm{y}}_{/\mathbbm{B}_{i}}\!\left(T\right)=0\right)-\mathbbm{P}^{2}\left(\hat{\bm{y}}_{/\mathbbm{B}_{i}}\!\left(T\right)=1\right)\right)\left(\mathbbm{E}\left[l\right]\right)^{2}\\
        &= \mathbbm{E}\left[l^{2}\right]+ \left(\mathbbm{P}\left(\hat{\bm{y}}_{/\mathbbm{B}_{i}}\!\left(T\right)=0\right)\right.\\
        &\left.\left(2\mathbbm{P}\left(\hat{\bm{y}}_{/\mathbbm{B}_{i}}\!\left(T\right)=1\right)-\mathbbm{P}\left(\hat{\bm{y}}_{/\mathbbm{B}_{i}}\!\left(T\right)=0\right)\right)\right.\\
        &\left.-\mathbbm{P}^{2}\left(\hat{\bm{y}}_{/\mathbbm{B}_{i}}\!\left(T\right)=1\right)\right)\left(\mathbbm{E}\left[l\right]\right)^{2}\\
        &\leq\mathbbm{E}\left[l^{2}\right]+\mathbbm{P}\left(\hat{\bm{y}}_{/\mathbbm{B}_{i}}\!\left(T\right)=0\right)\\
        &\left(2\mathbbm{P}\left(\hat{\bm{y}}_{/\mathbbm{B}_{i}}\!\left(T\right)=1\right)-\mathbbm{P}\left(\hat{\bm{y}}_{/\mathbbm{B}_{i}}\!\left(T\right)=\bm{0}\right)\right)\left(\mathbbm{E}\left[l\right]\right)^{2}\\
        &\leq\mathbbm{E}\left[l^{2}\right] +\mathbbm{P}\left(\hat{\bm{y}}_{/\mathbbm{B}_{i}}\!\left(T\right)=0\right)\\
        &\left(2\mathbbm{P}\left(\hat{\bm{y}}_{/\mathbbm{B}_{i}}\neq \bm{0}\right)-\mathbbm{P}\left(\hat{\bm{y}}_{/\mathbbm{B}_{i}}=\bm{0}\right)\right)\left(\mathbbm{E}\left[l\right]\right)^{2}\\
        &\leq\mathbbm{E}\left[l^{2}\right] +\\
        &\left(\begin{cases}
            \mathbbm{P}\left(\hat{\bm{y}}_{/\mathbbm{B}_{i}}=\bm{0}\right)\text{ if }f<0\text{,}\\
            1\text{ if }f\geq0\text{.}\\
        \end{cases}\right)\\
        &\left(2\mathbbm{P}\left(\hat{\bm{y}}_{/\mathbbm{B}_{i}}\neq \bm{0}\right)-\mathbbm{P}\left(\hat{\bm{y}}_{/\mathbbm{B}_{i}}=\bm{0}\right)\right)\left(\mathbbm{E}\left[l\right]\right)^{2}=\left(\sigma'\right)^{2}\text{,}
    \end{split}
    \label{eqn:sum_llr_var}
\end{equation}
where $f:=\left(2\mathbbm{P}\left(\hat{\bm{y}}_{/\mathbbm{B}_{i}}\neq \bm{0}\right)-\mathbbm{P}\left(\hat{\bm{y}}_{/\mathbbm{B}_{i}}=\bm{0}\right)\right)$.
As the decoding performance of the efficient hard-decision decoding is worse than the decoding performance of \gls{fht} decoding, the variance $\left(\sigma'\right)^{2}$ will be smaller when \gls{fht} decoding is used.
The error rate of using hard-decision decoding~\cite{saptharishi2016efficiently,saptharishi2017efficiently} is an upper bound on the error rate of using \gls{fht} decoding.

The last step of the aggregation is to sum the estimations from all different subspaces.
In~\cite{richardson2001deldpc,richardson2008modern}, the density function of the summation of (independent) random variables in the $l$-domain is derived by the convolution operation because the incoming messages have different check-node degrees; hence, the density functions are different.
For \gls{cpa} decoding, regardless of the subspace, the final estimation recovers the marginal probability according to~\eqref{eqn:aggr_marginalprob}; the same check-node operations (with the same degree) are applied to different subspaces; so the density function is the same for all subspaces.
Hence, instead of using the convolution operation, we assume that all subspaces are independent and use the central limit theorem to determine the density function of $\bar{\bm{l}}=\tfrac{\hat{\bm{l}}}{n_{\mathbbm{B}}} $:
\begin{equation}
     \sqrt{ n_{\mathbbm{B}}}\left(\bar{\bm{l}}(z)-\mathbbm{E}\left[u\right]\right) \sim \mathcal{N}\left(0,\text{Var}\left[u\right]\right)\text{,}
\end{equation}
where $\mathcal{N}\left(A,B\right)$ is the density function of a Gaussian distribution with a mean $A$ and a variance $B$;
\begin{equation}
    \sqrt{ n_{\mathbbm{B}}}\bar{\bm{l}}\left(z\right) \sim \mathcal{N}\left(\sqrt{n_{\mathbbm{B}}}\mathbbm{E}\left[u\right],\text{Var}\left[u\right]\right)\text{;}
\end{equation}
and
\begin{equation}
    \bar{\bm{l}}\left(z\right) \sim \mathcal{N}\left(\mathbbm{E}\left[u\right],\frac{\text{Var}\left[u\right]}{n_{\mathbbm{B}}}\right)\text{.}
\end{equation}
Therefore, the error probability is defined as
\begin{equation}
\begin{split}
    &\int_{-\infty}^{0}\mathrm{p}\left(\bar{\bm{l}}\left(z\right)\right) \dd{\bar{\bm{l}}\left(z\right)}\leq  \int_{-\infty}^{0}\mathcal{N}\left(\mu',\frac{\left(\sigma'\right)^{2}}{n_{\mathbbm{B}}}\right) \left(\bar{\bm{l}}\left(z\right)\right)\dd{\bar{\bm{l}}\left(z\right)}\text{.}
\end{split}
\end{equation}
Because only an upper bound is derived instead of the exact probability density function, in the analysis of this work, we assume that the \gls{llr} fed to the next iteration has a probability density function 
\begin{equation}
    \mathrm{p}\left(\bm{l}\left(z\right)\mid \bm{c}\left(z\right)=0\right)=\mathcal{N}\left(\mu',\frac{\left(\sigma'\right)^{2}}{n_{\mathbbm{B}}}\right)\left(\bm{l}\left(z\right)\right)\text{.}
    \label{eqn:llr_dist_next_iter}
\end{equation}
\subsection{Error Decay Rate of \gls{cpa} Decoding}
Given that the decoded soft information followed a Gaussian distribution, it is also interesting to know how fast the error rate decays in \gls{cpa} decoding.
We define another random variable $\mathsf{V}$ where its realization $v= \bm{l}\left(z\right) - \mu'$, we also assume that we are working on cases where $\mu'>0$, and then we have
\begin{equation}
    \begin{split}
    &\mathbbm{P}\left(\bm{l}\left(z\right)\leq0\right)\\
    &=\mathbbm{P}\left(\mathsf{V}\leq-\mu'\right)\\
    &=\int_{-\infty}^{-\mu'}\mathcal{N}\left(0,\frac{1}{n_{\mathbbm{B}}}\left(\sigma'\right)^{2}\right) \left(v\right)\dd{v}\\
    &=\int_{\mu'}^{\infty}\mathcal{N}\left(0,\frac{1}{n_{\mathbbm{B}}}\left(\sigma'\right)^{2}\right) \left(v\right)\dd{v}\\
    &=\frac{1}{\sqrt{2\pi}} \frac{\sqrt{n_{\mathbbm{B}}}}{\sigma'}\int_{\mu'}^{\infty}\exp\left(-\frac{v^{2}n_{\mathbbm{B}}}{2\left(\sigma'\right)^{2}}\right) \dd{v}\\
    &\stackrel{(a)}{\leq} \frac{1}{\sqrt{2\pi}} \frac{\sqrt{n_{\mathbbm{B}}}}{\sigma'}\int_{\mu'}^{\infty}\frac{v}{\mu'} \exp\left(-\frac{v^{2}n_{\mathbbm{B}}}{2\left(\sigma'\right)^{2}}\right)\dd{v}\\
    &=\frac{1}{\sqrt{2\pi}} \frac{\sqrt{n_{\mathbbm{B}}}}{\sigma'} \frac{1}{\mu'} \frac{\left(\sigma'\right)^{2}}{n_{\mathbbm{B}}} \exp\left(-\frac{\left(\mu'\right)^{2}n_{\mathbbm{B}}}{2\left(\sigma'\right)^{2}}\right)\\
    &=\frac{\sigma'}{\sqrt{2\pi}\mu'\sqrt{n_{\mathbbm{B}}}} \exp\left(-\frac{\left(\mu'\right)^{2}n_{\mathbbm{B}}}{2\left(\sigma'\right)^{2}}\right)\text{,}
    \end{split}
    \label{eqn:upper_bound_pa_bit_error_rate}
\end{equation}
where the inequity $(a)$ holds because $\tfrac{v}{\mu'}\geq1$~\cite{upper_bound_normal_distribution}.
Hence, we can conclude that, for a positive $\mu'$ and a bounded $\sigma'$, the error rate of \gls{cpa} decoding decays as 
\begin{equation}
    \propto \frac{1}{\sqrt{n_{\mathbbm{B}}}}\exp\left(-n_{\mathbbm{B}}\right)\text{.}
    \label{eqn:error_decay_num_subspaces}
\end{equation}
Also, the error rate decreases as $\mu'$ increases and $\sigma'$ decreases.

\subsection{Density Evolution and Numerical Simulations}
According to the exact marginal probability shown in~\eqref{eqn:aggr_marginalprob} and the symmetry (Proposition~\ref{prop:symmetric_CPA}), density evolution can analyze the density function of \gls{cpa} decoding.
We denote the convolution in the $g$-domain and then convert back to $l$-domain as $a\circledast b$ (i.e., check node update)~\cite{richardson2001deldpc,richardson2008modern}, where $g=\left(\text{lg}\left(l\right),-\ln\left(\left|\tanh\left(l\right)\right|\right)\right)$, $\text{lg}\left(l\right)=1$ if $l<0$, $\text{lg}\left(l\right)=0$ if $l>0$, and $\text{lg}\left(l\right)$ is $0$ or $1$ with a equal probability when $l=0$~\cite{richardson2001deldpc}.
Without loss of generality, assume all-zeros codewords are transmitted, and the $l$-density (i.e., the density function of the \gls{llr}) under the equiprobable and \gls{iid} source, the \gls{awgn} channel, and the \gls{bpsk} modulation is
\begin{equation}
    \begin{split}
     \mathbbm{P}\left(\bm{l}\left(z\right)\mid \bm{c}\left(z\right)=0\right)=\frac{2}{\sigma^{2}}\mathcal{N}\left(1, \sigma^{2}\right)= \mathcal{N}\left(\frac{2}{\sigma^{2}}, \frac{4}{\sigma^{2}}\right)\text{.}
    \end{split}
    \label{eqn:received_l_density}
\end{equation}
Hence, according to~\cite[Sec. 4.5.2]{richardson2008modern}, the results of the inverse hyperbolic tangent part of the aggregation function follow a distribution of $a_{i}^{\circledast 2^{r-1}-1}\left(l\right)$.
As shown in Section~\ref{sec:keyprop_CPA}, the probability of the correct decoding ($\hat{\bm{y}}_{/\mathbbm{B}_{i}}=\bm{0}$) of the \gls{ml} decoding is $\mathbbm{P}\left(\hat{\bm{y}}_{/\mathbbm{B}_{i}}=\bm{0}\right)$.
By the definition of the~\eqref{eqn:product_of_vars_aggreagtions_prob}, the density function returned from the aggregation function is
\begin{equation}
    \begin{cases}
        a_{i}^{\circledast 2^{r-1}}\left(l\right)\mathbbm{P}\left(\hat{\bm{y}}_{/\mathbbm{B}_{i}}\!\left(T\right)=0\right)\text{, }&\text{ if }\hat{\bm{y}}_{/\mathbbm{B}_{i}}\!\left(T\right)=0\text{,}\\
        a_{i}^{\circledast 2^{r-1}}\left(-l\right)\mathbbm{P}\left(\hat{\bm{y}}_{/\mathbbm{B}_{i}}\!\left(T\right)=1\right)\text{, }&\text{ if }\hat{\bm{y}}_{/\mathbbm{B}_{i}}\!\left(T\right)=1\text{.}\\
    \end{cases}
    \label{eqn:product_of_vars_aggreagtions_prob_de}
\end{equation}
According to the model in Section~\ref{sec:keyprop_CPA}, the central limit theorem defined by the mean~\eqref{eqn:sum_llr_mean} and variance~\eqref{eqn:sum_llr_var} is enough to analyze the distribution returned from \gls{cpa} decoding.
The density function for the LLR used in the next iteration is defined in~\eqref{eqn:llr_dist_next_iter}.

The pseudo-code of the density evolution of \gls{cpa} decoding is shown in Algorithm~\ref {alg:density_evolution_PA}.
We use the \gls{rm}$\left(7,3\right)$ code at $E_\mathrm{b}/N_{0}=2.5\text{ dB}$, and the optimized and \gls{pcpa} decoding with $128$ subspaces~\cite{Op_PCPA,Li2023improvePAlist} and the broadcast update~\cite{li2024layeredCPA} to run the simulations.
We first plot the number of occurrences of the received channel \gls{llr} (iter. $0$) and the \gls{llr} returned from each iteration of the \gls{pcpa} decoding in Fig.~\ref{fig:hist_pcpa}.
We can see that occurrences of the \gls{llr} roughly follow Gaussian distributions.
Hence, we plot the probability density function of the Gaussian distribution using the sample mean and variance of the \gls{llr} in Fig.~\ref{fig:pdf_pcpa}, which is denoted as ``sim.'' in the legend.
Also, the probability density function of the Gaussian distribution using the computed mean $\mu'$ and variance $\left(\sigma'\right)^{2}$ is also plotted in Fig.~\ref{fig:pdf_pcpa}.
The probability density function of the Gaussian distribution returned from our proposed density evolution has a similar trend to the trend of the histogram in Fig.~\ref{fig:hist_pcpa} and captures the mean- and variance-reduction feature, which explains the decoding mechanism behind \gls{cpa} decoding.

For the \gls{rm}$\left(8,3\right)$ code, the histogram of the \gls{llr}s returned from the \gls{pcpa} decoding with $256$ subspaces and the broadcast update is shown in Fig.~\ref{fig:hist_pcpa_83}, and the density function returned from the empirical simulation is shown in Fig.~\ref{fig:pdf_pcpa_83}. 
The \gls{llr} values are more concentrated around the mean value compared to the histogram in Fig.~\ref{fig:hist_pcpa} and the density functions in Fig.~\ref{fig:pdf_pcpa} for \gls{rm}$(7,3)$ codes with $128$ subspaces, which partially verifies the conjecture in~\eqref{eqn:error_decay_num_subspaces} where a larger the number of subspaces implies a faster decay speed of error rates.
Based on the saturation early stopping~\cite{RPA_BP,Li2023improvePAlist}, for two consecutive iterations, the difference in the L2 norm between returned soft information will be smaller than a threshold $\theta$, and hard-decision results will be the same after a couple of iterations, which triggers the early stopping and explains the fast convergence speed.
\begin{algorithm}[t]
\label{alg:normalizedFHT}
\footnotesize
\caption{\label{alg:density_evolution_PA} Density Evolution of \gls{cpa} Decoding}
    \DontPrintSemicolon
    \SetAlgoVlined
    \SetKwData{h}{$\bm{H}$}
    \SetKwData{g}{$\bm{G}$}
    \KwIn{Initial density function $\mathcal{N}\!\left(\frac{2}{\sigma^{2}}, \frac{4}{\sigma^{2}}\right)$, Code length $n$, Order parameter $r$, number of subspaces $n_{\mathbbm{B}}$}
    \KwOut{Output density $\mathrm{p}$}
    \SetKwFunction{len}{Len}
    \SetKwFunction{clip}{Clip}
    \SetKwFunction{rightshift}{ARightShift}
    $\mathrm{p}\left(\bm{l}\left(z\right)\mid \bm{c}\left(z\right)=0\right)\leftarrow \mathcal{N}\!\left(\frac{2}{\sigma^{2}}, \frac{4}{\sigma^{2}}\right)$\;
    $p\leftarrow \eqref{eqn:bit_error_prob}$\;
    \For{$i=1:N_{\text{max}}$}{
        $\Bar{p}\leftarrow\eqref{eqn:BSC_approx_proj_prob}$\;
        $\mathbbm{P}\left(\hat{\bm{y}}_{/\mathbbm{B}_{i}}=\bm{0}\right)\leftarrow\eqref{eqn:upper_bound_error_bsc}$\;
        $\mathbbm{P}\left(\hat{\bm{y}}_{/\mathbbm{B}_{i}}\neq\bm{0}\right)\leftarrow 1-\mathbbm{P}\left(\hat{\bm{y}}=\bm{0}\right)$\;
        $\mathrm{p}\left(l\right)\leftarrow a_{i}^{\circledast 2^{r-1}}\left(l\right)$\;
        $\mathbbm{E}\left[l\right]\leftarrow \int_{-\infty}^{+\infty}\mathrm{p}\left(l\right) l \dd l$\;
        $\mathbbm{E}\left[l^{2}\right]\leftarrow \int_{-\infty}^{+\infty}\mathrm{p}\left(l\right) l^{2} \dd l$\;
        $\mu^{'}\leftarrow\eqref{eqn:sum_llr_mean}$\;
        $\left(\sigma^{'}\right)^{2}\leftarrow \eqref{eqn:sum_llr_var}$\;
        $p\leftarrow \eqref{eqn:bit_error_prob}$\;
    }
    \KwRet{$\mathrm{p}\leftarrow\eqref{eqn:llr_dist_next_iter}$}
\end{algorithm}

\begin{figure}[h!]
    \centering
    \begin{tikzpicture}
  \begin{customlegend}[legend columns=3,legend style={at={(0.0,-5.0)},anchor=south west, align=left,draw=none},
  legend entries={
  \scriptsize{iter. $0$}, \scriptsize{iter. $1$}, \scriptsize{iter. $2$}
  }
  ]
  \addlegendimage{draw=colorblindfree11_1, very thick}
  \addlegendimage{draw=colorblindfree11_11, very thick}
  \addlegendimage{draw=colorblindfree11_3, very thick}
  
  \end{customlegend}
\end{tikzpicture}
\vspace{0.0em}

    \begin{tikzpicture}
\begin{axis}[
    no markers, 
    domain=-10:10,
    xmin=-8,
    xmax=8,
    samples=100,
    axis lines=left, 
    xlabel=LLR, 
    ybar=0.0pt,
    ylabel=\footnotesize Counts,
    height=.6\columnwidth, 
    width=\columnwidth,
    ymin=-1,
    ymax = 4e5,
    bar width=4pt,
    enlargelimits=upper,
    y tick label style={
        /pgf/number format/fixed zerofill},
    cycle list = {
         {draw=colorblindfree11_1, fill=colorblindfree11_1},
        {draw=colorblindfree11_11, fill=colorblindfree11_11},
        {draw=colorblindfree11_3, fill=colorblindfree11_3}
    }
]
 \addplot table[x=x, y=y0]{data/hist_llr_pcpa_73.txt};
\addplot table[x=x, y=y1]{data/hist_llr_pcpa_73.txt};
\addplot table[x=x, y=y2]{data/hist_llr_pcpa_73.txt};
\end{axis}
\end{tikzpicture}
    \caption{Histogram of \gls{llr}s from the channel and every iteration of the \gls{pcpa} decoding for the \gls{rm}$\left(7,3\right)$ code at a $E_\mathrm{b}/N_{0}=2.5\text{ dB}$ over $10^{4}$ frames.}
    \label{fig:hist_pcpa}
\end{figure}

\begin{figure}[h!]
    \centering
\begin{tikzpicture}
  \begin{customlegend}[legend columns=2,legend style={at={(0.0,-5.0)},anchor=south west, align=left,draw=none},
  legend entries={
  \scriptsize{$\mu=3.56, \sigma=2.67$, DE, iter. $0$}, 
  \scriptsize{$\mu=3.56, \sigma=2.67$, sim., iter. $0$}, 
  \scriptsize{$\mu=1.17, \sigma=0.14$, DE, iter. $1$}, 
  \scriptsize{$\mu=1.19, \sigma=0.31$, sim, iter. $1$}, 
  \scriptsize{$\mu=0.27, \sigma=5.3e^{-3}$, DE, iter. $2$}, 
  \scriptsize{$\mu=0.32, \sigma=0.15$, sim, iter. $2$}
  },
  width=\columnwidth
  ]
  \addlegendimage{draw=colorblindfree11_1, very thick}
  \addlegendimage{draw=colorblindfree11_1, very thick, densely dashed}
  \addlegendimage{draw=colorblindfree11_11, very thick}
  \addlegendimage{draw=colorblindfree11_11, very thick, densely dashed}
  \addlegendimage{draw=colorblindfree11_3, very thick}
  \addlegendimage{draw=colorblindfree11_3, very thick, densely dashed}
  
  \end{customlegend}
\end{tikzpicture}
\vspace{0.0em}

    \begin{tikzpicture}
\begin{axis}[
    no markers, 
    domain=-10:10,
    xmin=-8,
    xmax=8,
    samples=100,
    axis lines=left, 
    xlabel=LLR, 
    ylabel=$\mathcal{N}\left(\mu\text{, }\sigma\right)$,
    height=.6\columnwidth, 
    width=\columnwidth,
    ymin=0,
    ytick={0.0,0.2,...,1.0},
    enlargelimits=upper,
    y tick label style={
        /pgf/number format/fixed zerofill},
    cycle list = {
        {draw=colorblindfree11_1, very thick},
        {draw=colorblindfree11_1, ,mark=square*, very thick, densely dashed},
        {draw=colorblindfree11_11, very thick},
        {draw=colorblindfree11_11, ,mark=square*, very thick, densely dashed},
        {draw=colorblindfree11_3, very thick},
        {draw=colorblindfree11_3, ,mark=square*, very thick, densely dashed}
    }
]
  \addplot table [x=x, y=y0] {data/de_pdf_pcpa_73.txt};
  \addplot table [x=x, y=y0] {data/pdf_pcpa_73.txt};
  \addplot table [x=x, y=y1] {data/de_pdf_pcpa_73.txt};
  \addplot table [x=x, y=y1] {data/pdf_pcpa_73.txt};
  \addplot table [x=x, y=y2] {data/de_pdf_pcpa_73.txt};
  \addplot table [x=x, y=y2] {data/pdf_pcpa_73.txt};
\end{axis}
\end{tikzpicture}
    \caption{Density functions (sim.) of the output from the \gls{pcpa} decoding over $10^{4}$ frames and the density functions (DE) returned from the density evolution of the \gls{pcpa} decoding for the \gls{rm}$\left(7,3\right)$ code at a $E_\mathrm{b}/N_{0}=2.5\text{ dB}$.}
    \label{fig:pdf_pcpa}
\end{figure}

\begin{figure}[t]
    \centering
    \begin{tikzpicture}
  \begin{customlegend}[legend columns=3,legend style={at={(0.0,-5.0)},anchor=south west, align=left,draw=none},
  legend entries={
  \scriptsize{iter. $0$}, \scriptsize{iter. $1$}, \scriptsize{iter. $2$}
  }
  ]
  \addlegendimage{draw=colorblindfree11_1, very thick}
  \addlegendimage{draw=colorblindfree11_11, very thick}
  \addlegendimage{draw=colorblindfree11_3, very thick}
  
  \end{customlegend}
\end{tikzpicture}
\vspace{0.0em}

    \begin{tikzpicture}
\begin{axis}[
    no markers, 
    domain=-10:10,
    xmin=-8,
    xmax=8,
    samples=100,
    axis lines=left, 
    xlabel=LLR, 
    ybar=0.0pt,
    ylabel=\footnotesize Counts,
    height=.6\columnwidth, 
    width=\columnwidth,
    ymin=-1,
    ymax=3e6,
    bar width=5pt,
    enlargelimits=upper,
    y tick label style={
        /pgf/number format/fixed zerofill},
    cycle list = {
         {draw=colorblindfree11_1, fill=colorblindfree11_1},
        {draw=colorblindfree11_11, fill=colorblindfree11_11},
        {draw=colorblindfree11_3, fill=colorblindfree11_3}
    }
]
\addplot table[x=x, y=y0]{data/hist_llr_pcpa_83.txt};
\addplot table[x=x, y=y1]{data/hist_llr_pcpa_83.txt};
\addplot table[x=x, y=y2]{data/hist_llr_pcpa_83.txt};
\end{axis}
\end{tikzpicture}
    \caption{Histogram of \gls{llr}s from the channel and every iteration of the \gls{pcpa} decoding for the \gls{rm}$\left(8,3\right)$ code at a $E_\mathrm{b}/N_{0}=1.5\text{ dB}$ over $10^{4}$ frames.}
    \label{fig:hist_pcpa_83}
\end{figure}

\begin{figure}[h!]
    \centering
\begin{tikzpicture}
  \begin{customlegend}[legend columns=2,legend style={at={(0.0,-5.0)},anchor=south west, align=left,draw=none},
  legend entries={
  \scriptsize{$\mu=2.05, \sigma=2.03$, DE, iter. $0$}, 
  \scriptsize{$\mu=2.05, \sigma=2.03$, sim., iter. $0$}, 
  \scriptsize{$\mu=0.36, \sigma=0.06$, DE, iter. $1$}, 
  \scriptsize{$\mu=0.37, \sigma=0.19$, sim, iter. $1$}, 
  \scriptsize{$\mu=0.016, \sigma=2.6e^{-4}$, DE, iter. $2$}, 
  \scriptsize{$\mu=0.013, \sigma=0.011$, sim, iter. $2$}
  },
  width=\columnwidth
  ]
  \addlegendimage{draw=colorblindfree11_1, very thick}
  \addlegendimage{draw=colorblindfree11_1, very thick, densely dashed}
  \addlegendimage{draw=colorblindfree11_11, very thick}
  \addlegendimage{draw=colorblindfree11_11, very thick, densely dashed}
  \addlegendimage{draw=colorblindfree11_3, very thick}
  \addlegendimage{draw=colorblindfree11_3, very thick, densely dashed}
  
  \end{customlegend}
\end{tikzpicture}
\vspace{0.0em}

    \begin{tikzpicture}
\begin{axis}[
    no markers, 
    domain=-10:10,
    xmin=-8,
    xmax=8,
    samples=100,
    axis lines=left, 
    xlabel=LLR, 
    ylabel=$\mathcal{N}\left(\mu\text{, }\sigma\right)$,
    height=.6\columnwidth, 
    width=\columnwidth,
    ymin=0,
    ytick={0.0,0.2,...,1.0},
    enlargelimits=upper,
    y tick label style={
        /pgf/number format/fixed zerofill},
    cycle list = {
        {draw=colorblindfree11_1, very thick},
        {draw=colorblindfree11_1, ,mark=square*, very thick, densely dashed},
        {draw=colorblindfree11_11, very thick},
        {draw=colorblindfree11_11, ,mark=square*, very thick, densely dashed},
        {draw=colorblindfree11_3, very thick},
        {draw=colorblindfree11_3, ,mark=square*, very thick, densely dashed}
    }
]
  \addplot table [x=x, y=y0] {data/de_pdf_pcpa_83.txt};
  \addplot table [x=x, y=y0] {data/pdf_pcpa_83.txt};
  \addplot table [x=x, y=y1] {data/de_pdf_pcpa_83.txt};
  \addplot table [x=x, y=y1] {data/pdf_pcpa_83.txt};
  \addplot table [x=x, y=y2] {data/de_pdf_pcpa_83.txt};
  \addplot table [x=x, y=y2] {data/pdf_pcpa_83.txt};
\end{axis}
\end{tikzpicture}
    \caption{Density functions (sim.) of the output from the \gls{pcpa} decoding over $10^{4}$ frames and the density functions (DE) returned from the density evolution of the \gls{pcpa} decoding for the \gls{rm}$\left(8,3\right)$ code at a $E_\mathrm{b}/N_{0}=1.5\text{ dB}$.}
    \label{fig:pdf_pcpa_83}
\end{figure}

From Fig.~\ref{fig:pdf_pcpa}, we can see that our derived lower bound $\mu'$ is tight.
The variance $\left(\sigma'\right)^{2}/n_{\mathbbm{B}}$ is underestimated because, at a finite code length, the output \gls{llr}s at each iteration of the \gls{pcpa} decoding are not necessarily independent across all subspaces, and the summation of two (dependent) random variables $X_{1}$ and $X_{2}$ has a variance of $\sigma_{1}^{2}+\sigma_{2}^{2}+2\text{Cov}\left(X_{1},X_{2}\right)$ and this variance maybe larger than the variance $\sigma_{1}^{2}+\sigma_{2}^{2}$ of our assumption that variables are independent.
Hence, under the independent assumption and the proposed model, a useful bound for the decoding bit error~\eqref{eqn:upper_bound_pa_bit_error_rate} cannot be produced.
Also, the distribution of the summation of dependent random variables is hard to analyze, unlike the distributions of independent random variables, and it may not be a normal distribution~\cite{sum_dep_normal_may_not_be_normal}.
Furthermore, the interaction between dependent/correlated distributions returned from different subspaces is hard to analyze. 
In conclusion, a more fine-grained analysis (i.e., considering the dependency among subspaces) for the density evolution is needed for \gls{cpa} decoding in future work to support analysis such as the bound for the decoding error.

The underestimation effect on the variance $\left(\sigma'\right)^{2}/n_{\mathbbm{B}}$ is also reflected in the iterative density evolution analysis on the \gls{rm}$\left(8,3\right)$ codes.
From Fig.~\ref{fig:pdf_pcpa_83}, the density evolution analysis returns a lower mean value than the mean value returned from the empirical simulation at iteration $1$, which shows that our derived lower bound is tight for \gls{rm}$\left(8,3\right)$ codes, except for the mean value returned from the density evolution at iteration $2$.
It can be observed from Fig.~\ref{fig:hist_pcpa_83} that most \gls{llr}s are concentrated on several coarse-grained bins; hence, the empirical mean and the empirical standard deviation might not be accurate.

\section{Asymptotic Analysis of \gls{cpa} Decoding}
\label{sec:Asym_Ana_CPA}
In this section, the asymptotic analysis for \gls{cpa} decoding with one decoding iteration is performed.
The asymptotic analysis is based on the probabilistic model derived by our density evolution analysis for \gls{cpa} decoding, and, in this work, we conduct the asymptotic analysis based on the asymptotic behaviour for all key operations.
Combining the analytical results of the projection, the \gls{fht} decoding, and the aggregation function, we derive the asymptotic behaviour.

\subsection{Asymptotic Analysis of the Projection and \gls{fht} Decoding}
It is mentioned in~\cite{Sidelnikov,dumer2004recursive} that the ML decoding asymptotically decodes all but a vanishing fraction of error patterns of a weight up to $\tfrac{n}{2}\left(1-\epsilon_{r}^{\min}\right)$ for low-rate RM codes with a fixed $r$, and the residual term
\begin{equation}
    \lim_{m\rightarrow\infty}\epsilon_{r}^{\min}\left(m\right) = \lim_{m\rightarrow\infty}m^{r/2}n^{-1/2}\left(\frac{c\left(2^{r}-1\right)}{r!}\right)^{1/2}\text{,}
    \label{eqn:frac_nondecodeable_error_pattern}
\end{equation}
where $c$ is a constant.
For order-$1$ RM codes, the residual term becomes
\begin{equation}
    \begin{split}
    \lim_{m\rightarrow\infty}\epsilon_{1}^{\min}\left(m\right)=&\lim_{m\rightarrow\infty}m^{1/2}n^{-1/2}\left(\frac{c\left(2^{1}-1\right)}{1!}\right)^{1/2}\\
    &=\lim_{m\rightarrow\infty} \left(\frac{m}{2^{m}}\right)^{1/2}c^{1/2}\\
    &\stackrel{(a)}{=}\lim_{m\rightarrow\infty}\left(\frac{1}{m2^{m-1}}\right)^{1/2}c^{1/2}=0\text{,}
    \end{split}
    \label{eqn:limit_correctable_error_patterns_FHT}
\end{equation}
where $\left(a\right)$ in~\eqref{eqn:limit_correctable_error_patterns_FHT} holds by L'H\^opital's rule.
Hence, we can conclude that the \gls{fht} decoding, which is \gls{ml} decoding for order-$1$ RM codes, can decode all but a vanishing fraction of error patterns of weights up to $\tfrac{n}{2}\left(1-\lim_{m\rightarrow\infty}\epsilon_{1}^{\min}\right)=\tfrac{n}{2}$.
Then, the probability of receiving soft information that can be correctly decoded by the \gls{fht} decoding is
\begin{equation}
    \sum_{i=0}^{n/2}\binom{n}{i}\hat{p}^{i}\left(1-\hat{p}\right)^{n-i}\text{,}
\end{equation}
which is the cumulative probability of having at most $\tfrac{n}{2}$ negative soft information, and $\hat{p}$ is the error probability of projected soft information.
In this work, we are interested in the asymptotic behaviour (i.e., $n\rightarrow\infty$) of received sequences that can be correctly decoded.

We can define a variable $\mathsf{X}=\sum_{z\in\{0,1\}^{m-r+1}}\bm{x}\left(z\right)$ where $\bm{x}\left(z\right)=\mathbbm{1}\!\left(\bm{l}_{/\mathbbm{B}_{i}}\left(z\right)<0\right)$ and the $z$ is the index for the projected bit estimation.
The random variable $\bm{x}\left(z\right)$ takes value $1$ with a probability of $\hat{p}$ and value $0$ with a probability of $1-\hat{p}$.
By the central limit theorem, the distribution of $\mathsf{X}$ can be well approximated (normal approximation to the binomial distribution) by $\mathbbm{P}\left(\mathsf{X}\right)=\mathcal{N}\left(n\hat{p},n\hat{p}\left(1-\hat{p}\right)\right)$.
The random variable $\mathsf{X}=i$ is exactly the event where $i$ out of $n$ projected code bits take the value $1$, which are the erroneous code bits.
Hence, 
\begin{equation}
    \sum_{i=0}^{n/2}\binom{n}{i}\hat{p}^{i}\left(1-\hat{p}\right)^{n-i} \approx \mathbbm{P}\left(\mathsf{X}\leq \frac{n}{2}\right)
\end{equation}
as $n\rightarrow\infty$.
By defining the standardized variable 
\begin{equation}
    \mathsf{Z}=\frac{\mathsf{X}-n\hat{p}}{\sqrt{n\hat{p}\left(1-\hat{p}\right)}}
\end{equation}
and $\mathbbm{P}\left(\mathsf{Z}\right)=\mathcal{N}\left(0,1\right)$, we have
\begin{equation}
    \mathbbm{P}\left(\mathsf{X}\leq \frac{n}{2}\right)=\mathbbm{P}\left(\mathsf{Z}\leq\frac{\sqrt{n}\left(1/2-\hat{p}\right)}{\sqrt{\hat{p}\left(1-\hat{p}\right)}}\right)\text{,}
\end{equation}
and
\begin{equation}
    \lim_{n\rightarrow\infty} \mathbbm{P}\left(\mathsf{Z}\leq\frac{\sqrt{n}\left(1/2-\hat{p}\right)}{\sqrt{\hat{p}\left(1-\hat{p}\right)}}\right)=
    \begin{cases}
        1\text{, if }\hat{p}<\frac{1}{2}\text{,}\\
        0\text{, if }\hat{p}>\frac{1}{2}\text{,}\\
        \frac{1}{2}\text{, if }\hat{p}=\frac{1}{2}\text{.}\\
    \end{cases}
    \label{eqn:asymptotic_FHT_block_error}
\end{equation}

We assume that we are working on RM codes with a code rate $R\leq C$, where $C$ is the channel capacity.
Let $r'$ be the solution of
\begin{equation}
    C\geq\max_{r'\in \{0,1,...,m\}}\left(\sum_{i=0}^{r'} \binom{m}{i}\right)\frac{1}{n}\text{,}
\end{equation}
and we assume that we are working on RM codes with $r\leq r'$.
When $r=m$, we have a perfect channel condition with a signal-to-noise ratio $=\infty$, and the vanishing error probability is trivially achieved.
Hence, we focus on the discussion on the finite signal-to-noise ratio, and $r<m$.

Based on the projection function, we can see that the \gls{llr} is $<0$ when there is an odd number of received \glspl{llr} in $\bm{l}$ is $<0$.
The probability $\hat{p}$ of projected \glspl{llr} $<0$, by~\eqref{eqn:BSC_approx_proj_prob}, is 
\begin{equation}
\begin{split}
    \hat{p}
    &=\frac{1}{2}\left(1-\left(1-2p\right)^{2^{r-1}}\right)\text{,}
\end{split}
\end{equation}
\begin{equation}
    p=\int_{-\infty}^{0}\frac{1}{\sqrt{2\pi \sigma^{2}}} \exp(-\frac{(y-1)^{2}}{2\sigma^{2}}) \dd{y}\text{,}
\end{equation}
and $y$ is the received symbol from the channel.
Under the assumption of a bounded variance in the \gls{awgn} channel, we have $p<1/2$.
When $p<1/2$, the probability of a projected \gls{llr} $<0$ is $\hat{p}<1/2$ when the order parameter $r<\infty$.
Hence, the asymptotically achievable code rate is
\begin{equation}
\begin{split}
    R = \lim_{m\rightarrow\infty}\frac{\sum_{i=0}^{r<\infty}\binom{m}{i}}{2^{m}}
    \leq\sum_{i=0}^{r<\infty}\frac{1}{i!}\lim_{m\rightarrow\infty}\frac{m^{i}}{2^{m}}=0\text{.}
\end{split}
\end{equation}
Also, $\lim_{m\rightarrow\infty}m-r+1=\infty$ for $r<\infty$, \eqref{eqn:limit_correctable_error_patterns_FHT} holds for the order-$1$ sub-codes for \gls{cpa} decoding, and \gls{fht} decoding decodes all but a vanishing fraction of error patterns by~\eqref{eqn:asymptotic_FHT_block_error} because $\hat{p}<1/2$.

\subsection{Asymptotic Analysis on the Aggregation Function}
In the following theorem, we show that the random variable $l$ has a positive mean and bounded variance when $r<\infty$ and the input \gls{llr} has a positive mean value.
\begin{theorem}
    Assume the density of the \gls{llr} has the symmetric property~\cite{richardson2008modern} where
\begin{equation}
    \mathrm{p}\left(l\right) = \mathrm{p}\left(-l\right)\exp\left(l\right)\text{,}
\end{equation}
and the random variable 
\begin{equation}
\begin{split}
    d:=\tanh\left(\frac{l}{2}\right)=\prod_{i=1}^{2^{r-1}-1}\tanh\left(\frac{l_{i}}{2}\right)
\end{split}
    \label{eqn:d_density_def}
\end{equation}
has a symmetric $d$-density~\cite{richardson2008modern} where
\begin{equation}
    \mathrm{p}\left(d\right) = \mathrm{p}\left(-d\right)\frac{1+d}{1-d}\text{.}
\end{equation}
   Two density functions can be converted by~\cite{richardson2008modern}
\begin{equation}
    \mathrm{p}\left(l\right) = \frac{\mathrm{p}\left(d\right)}{2\cosh^{2}\left(\frac{l}{2}\right)}\text{.}
    \label{eqn:d_to_l_density)func}
\end{equation} 
    $\mathbbm{E}[d]>0$ if and only if $\mathbbm{E}[l]>0$ when the density function of $d$ is symmetric.
    \label{thm:positive_mean_aggregation}
\end{theorem}
\begin{proof}
    \textbf{Necessary condition:} 
By definition, we have~\cite{richardson2008modern}
\begin{equation}
\begin{split}
    &\mathbbm{E}\left[\tanh\left(\frac{l}{2}\right)\right]\\
    &=\int_{-\infty}^{\infty}\mathrm{p}\left(\tanh\left(\frac{l}{2}\right)\right)\tanh\left(\frac{l}{2}\right)\dd{l}\\
    &=\int_{-\infty}^{0}\mathrm{p}\left(\tanh\left(\frac{l}{2}\right)\right)\tanh\left(\frac{l}{2}\right)\dd{l}\\
    &+ \int_{0}^{\infty}\mathrm{p}\left(\tanh\left(\frac{l}{2}\right)\right)\tanh\left(\frac{l}{2}\right)\dd{l}\\
    &=-\int_{0}^{\infty}\frac{1-\tanh\left(\frac{l}{2}\right)}{1+\tanh\left(\frac{l}{2}\right)}\mathrm{p}\left(\tanh\left(\frac{l}{2}\right)\right)\tanh\left(\frac{l}{2}\right)\dd{l}\\
    &+ \int_{0}^{\infty}\mathrm{p}\left(\tanh\left(\frac{l}{2}\right)\right)\tanh\left(\frac{l}{2}\right)\dd{l}\\
    &=\int_{0}^{\infty}\left(1-\frac{1-\tanh\left(\frac{l}{2}\right)}{1+\tanh\left(\frac{l}{2}\right)}\right) \mathrm{p}\left(\tanh\left(\frac{l}{2}\right)\right) \tanh\left(\frac{l}{2}\right) \dd{l}\\
    &= \sum_{i=1}^{\infty} \int_{a_{i}}^{b_{i}}\frac{2\tanh\left(\frac{l}{2}\right)}{1+\tanh\left(\frac{l}{2}\right)} \mathrm{p}\left(\tanh\left(\frac{l}{2}\right)\right) \tanh\left(\frac{l}{2}\right) \dd{l}>0\text{,}
\end{split}
\end{equation}
where $\bigcup_{i=1}^{\infty}\left(a_{i},b_{i}\right)=\left[0,\infty\right)$, $a_{i},b_{i}\in\mathbbm{R}$, $a_{i}<b_{i}$, $b_{i}\leq a_{i+1}$, and $0\leq\tanh\left(\frac{l_{i}}{2}\right)<1\;\forall \text{ }l_{i}\in\left[0,\infty\right)$.
This result implies that
\begin{equation}
    \frac{2\tanh\left(\frac{l}{2}\right)}{1+\tanh\left(\frac{l}{2}\right)} \mathbbm{P}\left(\tanh\left(\frac{l}{2}\right)\right)>0\text{.}
    \label{eqn:positive_d_expected_value}
\end{equation}

Then, we have
\begin{equation}
\begin{split}
    &\mathbbm{E}\left[l\right]\\
    &= \int_{-\infty}^{\infty} \mathrm{p}\left(l\right)l \dd{l}\\
    &=\int_{-\infty}^{\infty} \mathrm{p}\left(\tanh\left(\frac{l}{2}\right)\right) \frac{l}{2\cosh^{2}\left(\frac{l}{2}\right)} \dd{l}\\
    &=\int_{0}^{\infty} \frac{2\tanh\left(\frac{l}{2}\right)}{1+\tanh\left(\frac{l}{2}\right)} \mathrm{p}\left(\tanh\left(\frac{l}{2}\right)\right) \frac{l}{2\cosh^{2}\left(\frac{l}{2}\right)} \dd{l}\\
    &=\sum_{i=1}^{\infty} \int_{a_{i}}^{b_{i}} \frac{2\tanh\left(\frac{l}{2}\right)}{1+\tanh\left(\frac{l}{2}\right)} \mathrm{p}\left(\tanh\left(\frac{l}{2}\right)\right) \frac{l}{2\cosh^{2}\left(\frac{l}{2}\right)} \dd{l}>0
\end{split}
\label{eqn:sym_l_density_mean_bound}
\end{equation}
because $\frac{l}{2\cosh^{2}\left(\frac{l}{2}\right)}>0\text{ }\forall \text{ }l\in\left(0,\infty\right)$ just like $\tanh\left(\frac{l}{2}\right)$ and $\exists\left(a_{i},b_{i}\right)\subseteq\left[0,\infty\right)$ such that~\eqref{eqn:positive_d_expected_value} holds.

\textbf{Sufficient condition:} The necessary condition can be prove similarly because $\mathbbm{E}[l]>0$ also implies $\exists\left(a_{i},b_{i}\right)\subseteq\left[0,\infty\right)$, such that~\eqref{eqn:positive_d_expected_value} holds, which leads to $\mathbbm{E}[d]>0$.
\end{proof}

By the sufficient condition of Theorem~\ref{thm:positive_mean_aggregation}, the expected value
\begin{align*}
    \mathbbm{E}\left[\tanh\left(\frac{\bm{l}\left(z\right)}{2}\right)\right]>0
\end{align*}
when the input \gls{llr} has a positive mean $\mathbbm{E}\left[\bm{l}\left(z\right)\right]>0$.
Because of the \gls{iid} assumption on the transmitted code bits, it is shown in~\cite{saeyoung2001Gaussianapproximation} that
\begin{equation}
\begin{split}
    \mathbbm{E}\left[\prod_{i=1}^{2^{r-1}-1}\tanh\left(\frac{l_{i}}{2}\right)\right] &= \mathbbm{E}\left[\tanh\left(\frac{l_{i}}{2}\right)\right]^{2^{r-1}-1}>0
\end{split}
\end{equation}
when $r<\infty$.
Hence, by Theorem~\ref{thm:positive_mean_aggregation}, $\mathbbm{E}[l]>0$ when $r<\infty$.

According to~\eqref{eqn:d_to_l_density)func}, the $d$-density can be converted to the $l$-density.
Hence, the second-order moment of $l=2\tanh^{-1}\left(d\right)$ can be computed by
\begin{equation}
\begin{split}
    &\mathbbm{E}\left[l^{2}\right]\\
    &=\int_{-\infty}^{\infty} \mathrm{p}\left(l\right) l^{2} \dd{l}\\
    &=\int_{-\infty}^{\infty} \mathrm{p}\left(\tanh\left(\frac{l}{2}\right)\right) \frac{l^{2}}{2\cosh^{2}\left(\frac{l}{2}\right)}  \dd{l}\\
    &\stackrel{(a)}{\leq} \int_{-\infty}^{\infty} \mathrm{p}\left(\tanh\left(\frac{l}{2}\right)\right) 4 \dd{l}=4\text{,}\\
\end{split}
    \label{eqn:l_density_with_d_density_variance_bound}
\end{equation}
the inequality $(a)$ in~\eqref{eqn:l_density_with_d_density_variance_bound} holds because
\begin{equation}
\begin{split}
    &\frac{l^{2}}{2\cosh^{2}\left(l\right)}\\
    &= \frac{l^{2}}{2\sinh^{2}\left(\frac{l}{2}\right)}2\tanh^{2}\left(\frac{l}{2}\right)\\
    &=\frac{2l^{2}}{\left(\exp\left(\frac{l}{2}\right) - \exp\left(-\frac{l}{2}\right)\right)^{2}}2\tanh^{2}\left(\frac{l}{2}\right)\\
    &\stackrel{(b)}{\leq}\frac{2\left(\exp\left(\frac{l}{2}\right) + \exp\left(-\frac{l}{2}\right)\right)^{2}}{\left(\exp\left(\frac{l}{2}\right) - \exp\left(-\frac{l}{2}\right)\right)^{2}}2\tanh^{2}\left(\frac{l}{2}\right)\\
    &=\frac{1}{\tanh^{2}\left(\frac{l}{2}\right)} 4\tanh^{2}\left(\frac{l}{2}\right)=4\text{,}
\end{split}
\end{equation}
and the inequality $(b)$ holds because
\begin{equation}
\begin{split}
    \left(\exp\left(\frac{l}{2}\right) + \exp\left(-\frac{l}{2}\right)\right)^{2}&= \left(2\cosh\left(\frac{l}{2}\right)\right)^{2}\\
    &\stackrel{(c)}{=} 4\left(\sum_{i=0}^{\infty}\frac{\left(l/2\right)^{2i}}{\left(2i\right)!}\right)^{2}\\
    &\stackrel{(d)}{\geq} 4 \left(1+\frac{l^{2}}{8}\right)^{2}\geq l^{2}\text{,}
\end{split}
\end{equation}
where the equality $(c)$ holds by the definition of the Taylor series for the $\cosh\left(\cdot\right)$ function, and the inequality $(d)$ holds by keeping only the first two terms of the series and truncating the rest.
The bounded second-order moment ($\mathbbm{E}\left[l^{2}\right]$) implies the first-order moment (i.e., expected value $\mathbbm{E}\left[l\right]$) is also bounded.
Hence,
\begin{equation}
\begin{split}
    &\text{Var}\left[l\right]= \mathbbm{E}\left[l^{2}\right] - \left(\mathbbm{E}\left[l\right]\right)^{2}\leq 4 - \left(\mathbbm{E}\left[l\right]\right)^{2} < \infty\text{.}
\end{split}
\label{eqn:sym_l_density_aggre_l_var_bound}
\end{equation}
Hence, when $r<\infty$, the positive mean~\eqref{eqn:sym_l_density_mean_bound}, the bounded variance~\eqref{eqn:sym_l_density_aggre_l_var_bound}, and the asymptotically vanishing error probability in the \gls{fht} decoding~\eqref{eqn:asymptotic_FHT_block_error} are returned from \gls{cpa} decoding, and we can see that the random variable $u$, which is defined in~\eqref{eqn:aggr_var_def}, returned from the estimation for the aggregation based on a subspace $\mathbbm{B}_{i}$ asymptotically has
\begin{equation}
    \mathbbm{E}\left[u\right]\stackrel{(a)}{=}\mathbbm{E}\left[l\right]\text{, }\text{Var}\left[u\right]\stackrel{(b)}{=}\text{Var}\left[l\right]\text{,}
\end{equation}
and the equality $(a)$ holds because of~\eqref{eqn:expected_value_u} and the equality $(b)$ holds because of~\eqref{eqn:sum_llr_var}.
If the number of subspaces is asymptotically infinite, then, by~\eqref{eqn:llr_dist_next_iter}, the vanishing error probability can be achieved.

\subsection{Asymptotic Analysis on the Number of Subspaces and the Asymptotic Behaviour of \gls{cpa} Decoding}
From~\eqref{eqn:llr_dist_next_iter}, we know that the averaging over the summation of the aggregation results is a form of variance reduction.
Hence, if the number of subspaces $\mathbbm{B}_{i}$ goes to $\infty$ as $n$ goes to $\infty$, \gls{cpa} decoding will asymptotically achieve vanishing error probability.
The following proposition shows that the number of subspaces can go to infinity as $n\rightarrow\infty$.
\begin{prop}
    If $1<r\leq \tfrac{m}{q} +1$, then $\lim_{m\rightarrow\infty}n_{\mathbbm{B}}=\lim_{m\rightarrow\infty}\binom{m}{r-1}_{2}=\infty$ for all $q>1$.
    \label{prop:limit_num_subspaces}
\end{prop}
\begin{proof}
    The Gaussian binomial coefficient is defined as
    \begin{equation}
        \binom{m}{K}_{2}=\frac{\prod_{i=0}^{K-1}\left(2^{m}-2^{i}\right)}{\prod_{j=0}^{K-1}\left(2^{K}-2^{j}\right)} > \frac{\left(2^{m}-2^{K}\right)^{K}}{\left(2^{K}\right)^{K}}\text{.}
    \end{equation}
    When $m\rightarrow\infty$, we have
    \begin{equation}
    \begin{split}
        &\lim_{m\rightarrow\infty}\binom{m}{K}_{2}\\
        &>\lim_{m\rightarrow\infty}\frac{\left(2^{m}-2^{K}\right)^{K}}{\left(2^{K}\right)^{K}}\\
        &=\lim_{m\rightarrow\infty}\left(2^{m\left(1-\frac{K\left(m\right)}{m}\right)}-1\right)^{K}\\
        &=\begin{cases}
            \infty\text{, }&\text{if }0< K\left(m\right)<m\text{ or }K'\left(m\right)<1\text{,}\\
            0\text{, }&\text{if }K\left(m\right)=m \text{ or }K'\left(m\right)=1\text{,}\\
            \left(-1\right)^{\lim_{m\rightarrow\infty}K\left(m\right)}\text{, }&\text{if }K'\left(m\right)>1\text{,}\\
            1\text{, }&\text{if }K=0\text{,}
        \end{cases}
    \end{split}
    \end{equation}
    where $0\leq K\left(m\right)\leq m$ denotes the parameter $K$ is a function of $m$, and $K'\left(m\right)$ is the derivative with respect to $m$.
    Hence, we can define $K\left(m\right):=\tfrac{m}{q}$ for $q>1$.
    Let $0< K=r-1\leq \tfrac{m}{q}$, we have $1< r\leq \tfrac{m}{q} +1<m+1$.
    The maximum growth rate of $K\left(m\right)$ is achieved when $1<q<2$.
\end{proof}
The largest code rate that has an infinite number of subspaces can be computed as follows.
The code rate of the \gls{rm}$\left(m,\tfrac{m}{q} +1\right)$ codes can be computed by
\begin{equation}
\begin{split}
    \frac{\sum_{i=0}^{m/q+1}\binom{m}{i}}{2^{m}}&=\sum_{i=0}^{m/q+1}\binom{m}{i}0.5^{i}0.5^{m-i}\\
    &=\mathbbm{P}\left(i\leq \frac{m}{q}+1\right)\text{,}
\end{split}
    \label{eqn:std_Chebyshev}
\end{equation}
which is the distribution function of the binomial distribution with a flipping probability of $0.5$, a mean value of $\tfrac{m}{2}$, and a variance of $0.5^{2}m$.
By the definition of the cumulative distribution, asymptotically, the largest code-rate, which has an infinite number of subspaces, is
\begin{equation}
    \begin{split}
        &\lim_{m\rightarrow\infty}\mathbbm{P}\left(i\leq\frac{m}{q}+1\right)=\mathbbm{P}\left(i\leq\infty\right)=1\text{.}
    \end{split}
    \label{eqn:limit_code_rate_cpa}
\end{equation}

Given the analysis on the projection, \gls{fht} decoding, and the aggregation function $\left(r<m=\infty\right)$, the limiting code rate ($R=$~\eqref{eqn:limit_code_rate_cpa}), and the given channel capacity $C$, we can conclude that, when code rate $R\leq \min\left\{C,1,0\right\}=0$, \gls{cpa} decoding can asymptotically achieve a vanishing error probability under the vanishing code rate.

\section{Conclusion}
\label{sec:conclusion}
We prove that \gls{cpa} decoding returns the exact marginal probability and is symmetric.
Then, we build a density evolution model to analyze \gls{cpa} decoding.
Simulation results show that our proposed density evolution model captures the fast reduction in the mean and the variance of the soft information returned from \gls{cpa} decoding, and these results qualitatively explain the decoding mechanism and the fast convergence behind the CPA decoding. 
Lastly, we perform an asymptotic analysis on \gls{cpa} decoding based on the proposed density evolution model, and we find that \gls{cpa} decoding asymptotically achieves a vanishing error probability when decoding \gls{rm} codes with a vanishing code rate.
The analysis in this work provides tools and insights for designing soft-decision \gls{pa}-based decoding with reduced complexity and improved decoding performance in future work.

\FloatBarrier

\bibliographystyle{IEEEtran}
\bibliography{IEEEabrv,reference.bib}{}

\appendices
\end{document}

%% file: color_package.tex
\usepackage[dvipsnames]{xcolor}

\definecolor{colorblindfree3_1}{RGB}{252,141,89}
\definecolor{colorblindfree3_2}{RGB}{255,255,191}
\definecolor{colorblindfree3_3}{RGB}{145,191,219}

\definecolor{colorblindfree4_1}{RGB}{215,25,28}
\definecolor{colorblindfree4_2}{RGB}{253,174,97}
\definecolor{colorblindfree4_3}{RGB}{171,217,233}
\definecolor{colorblindfree4_4}{RGB}{44,123,182}

\definecolor{colorblindfree5_1}{RGB}{215,25,28}
\definecolor{colorblindfree5_2}{RGB}{253,174,97}
\definecolor{colorblindfree5_3}{RGB}{255,255,191}
\definecolor{colorblindfree5_4}{RGB}{171,217,233}
\definecolor{colorblindfree5_5}{RGB}{44,123,182}

\definecolor{colorblindfree6_1}{RGB}{215,48,39}
\definecolor{colorblindfree6_2}{RGB}{252,141,89}
\definecolor{colorblindfree6_3}{RGB}{254,224,144}
\definecolor{colorblindfree6_4}{RGB}{224,243,248}
\definecolor{colorblindfree6_5}{RGB}{145,191,219}
\definecolor{colorblindfree6_6}{RGB}{69,117,180}

\definecolor{colorblindfree7_1}{RGB}{215,48,39}
\definecolor{colorblindfree7_2}{RGB}{244,109,67}
\definecolor{colorblindfree7_3}{RGB}{253,174,97}
\definecolor{colorblindfree7_4}{RGB}{254,224,144}
\definecolor{colorblindfree7_5}{RGB}{224,243,248}
\definecolor{colorblindfree7_6}{RGB}{171,217,233}
\definecolor{colorblindfree7_7}{RGB}{116,173,209}

\definecolor{colorblindfree8_1}{RGB}{215,48,39}
\definecolor{colorblindfree8_2}{RGB}{244,109,67}
\definecolor{colorblindfree8_3}{RGB}{253,174,97}
\definecolor{colorblindfree8_4}{RGB}{254,224,144}
\definecolor{colorblindfree8_5}{RGB}{224,243,248}
\definecolor{colorblindfree8_6}{RGB}{171,217,233}
\definecolor{colorblindfree8_7}{RGB}{116,173,209}
\definecolor{colorblindfree8_8}{RGB}{69,117,180}

\definecolor{colorblindfree9_1}{RGB}{215,48,39}
\definecolor{colorblindfree9_2}{RGB}{244,109,67}
\definecolor{colorblindfree9_3}{RGB}{253,174,97}
\definecolor{colorblindfree9_4}{RGB}{254,224,144}
\definecolor{colorblindfree9_5}{RGB}{255,255,191}
\definecolor{colorblindfree9_6}{RGB}{224,243,248}
\definecolor{colorblindfree9_7}{RGB}{171,217,233}
\definecolor{colorblindfree9_8}{RGB}{116,173,209}
\definecolor{colorblindfree9_9}{RGB}{69,117,180}

\definecolor{colorblindfree10_1}{RGB}{165,0,38}
\definecolor{colorblindfree10_2}{RGB}{215,48,39}
\definecolor{colorblindfree10_3}{RGB}{244,109,67}
\definecolor{colorblindfree10_4}{RGB}{253,174,97}
\definecolor{colorblindfree10_5}{RGB}{254,224,144}
\definecolor{colorblindfree10_6}{RGB}{224,243,248}
\definecolor{colorblindfree10_7}{RGB}{171,217,233}
\definecolor{colorblindfree10_8}{RGB}{116,173,209}
\definecolor{colorblindfree10_9}{RGB}{69,117,180}
\definecolor{colorblindfree10_10}{RGB}{49,54,149}

\definecolor{colorblindfree11_1}{RGB}{165,0,38}
\definecolor{colorblindfree11_2}{RGB}{215,48,39}
\definecolor{colorblindfree11_3}{RGB}{244,109,67}
\definecolor{colorblindfree11_4}{RGB}{253,174,97}
\definecolor{colorblindfree11_5}{RGB}{254,224,144}
\definecolor{colorblindfree11_6}{RGB}{255,255,191}
\definecolor{colorblindfree11_7}{RGB}{224,243,248}
\definecolor{colorblindfree11_8}{RGB}{171,217,233}
\definecolor{colorblindfree11_9}{RGB}{116,173,209}
\definecolor{colorblindfree11_10}{RGB}{69,117,180}
\definecolor{colorblindfree11_11}{RGB}{49,54,149}

%% file: IEEEabrv.bib
@STRING{IEEE_J_COML       = "{IEEE} Commun. Lett."}

@STRING{IEEE_J_COM        = "{IEEE} Trans. Commun."}

@STRING{IEEE_J_IT         = "{IEEE} Trans. Inform. Theory"}


%% file: reference.bib
@ARTICLE{RMcode,
  author={Muller, D. E.},
  journal={Trans. of the I.R.E. Professional Group on Electronic Computers}, 
  title={Application of Boolean algebra to switching circuit design and to error detection}, 
  year={1954},
  volume={EC-3},
  number={3},
  pages={6-12},
  doi={10.1109/IREPGELC.1954.6499441}}

@INPROCEEDINGS{rmpolar,
  author={Ar{\i}kan, Erdal},
  booktitle={IEEE Information Theory Workshop on Information Theory (ITW 2010, Cairo)}, 
  title={A survey of {Reed-Muller} codes from polar coding perspective}, 
  year={2010},
  volume={},
  number={},
  pages={1-5},
  doi={10.1109/ITWKSPS.2010.5503223}}

@ARTICLE{RMBEC,
  author={Kudekar, Shrinivas and Kumar, Santhosh and Mondelli, Marco and Pfister, Henry D. and Şaşoǧlu, Eren and Urbanke, Rüdiger L.},
  journal=IEEE_J_IT, 
  title={ {Reed–Muller} Codes Achieve Capacity on Erasure Channels}, 
  year={2017},
  volume={63},
  number={7},
  pages={4298-4316},
  doi={10.1109/TIT.2017.2673829}}

@ARTICLE{RMBSC,
  author={Abbe, Emmanuel and Shpilka, Amir and Wigderson, Avi},
  journal=IEEE_J_IT, 
  title={ {Reed–Muller} Codes for Random Erasures and Errors}, 
  year={2015},
  volume={61},
  number={10},
  pages={5229-5252},
  doi={10.1109/TIT.2015.2462817}}

@inproceedings{RMBSCimprove,
  author = {Sberlo, Ori and Shpilka, Amir},
  title = {On the Performance of {Reed-Muller} Codes with Respect to Random Errors and Erasures},
  year = {2020},
  publisher = {Society for Industrial and Applied Mathematics},
  address = {USA},
  booktitle = {Proc. of the Thirty-First Annual ACM-SIAM Symposium on Discrete Algorithms},
  pages = {1357–1376},
  numpages = {20},
  location = {Salt Lake City, Utah},
  series = {SODA '20}
}

@ARTICLE{RMBMS_bit_journal,
  author={Reeves, Galen and Pfister, Henry D.},
  journal=IEEE_J_IT,
  title={{Reed–Muller} Codes on {BMS} Channels Achieve Vanishing Bit-Error Probability for All Rates Below Capacity}, 
  year={2023},
  volume={},
  number={},
  pages={1-1},
  doi={10.1109/TIT.2023.3286452}}

@INPROCEEDINGS{RMBMS_block,
  author={Abbe, Emmanuel and Sandon, Colin},
  booktitle={IEEE 64th Annual Symposium on Foundations of Computer Science (FOCS)}, 
  title={A proof that {Reed-Muller} codes achieve {Shannon} capacity on symmetric channels}, 
  year={2023},
  volume={},
  number={},
  pages={177-193},
  doi={10.1109/FOCS57990.2023.00020}}

@ARTICLE{rmpolarize,
  author={Abbe, Emmanuel and Ye, Min},
  journal=IEEE_J_IT, 
  title={ {Reed-Muller} Codes Polarize}, 
  year={2020},
  volume={66},
  number={12},
  pages={7311-7332},
  doi={10.1109/TIT.2020.3023487}}

@ARTICLE{reed,
  author={Reed, I.},
  journal={Trans. of the IRE Professional Group on Information Theory},
  title={A class of multiple-error-correcting codes and the decoding scheme}, 
  year={1954},
  volume={4},
  number={4},
  pages={38-49},
  doi={10.1109/TIT.1954.1057465}}

@article{FHT,
  author        = "R. R. Green",
  title         = "A serial orthogonal decoder",
  journal       = "JPL Space Programs Summary",
  volume        = "37",
  year          = "1966",
  pages         = "247-253"
}

@article{FHT_soft,
  author={Y. {Be'ery} and J. {Snyders}},
  journal=IEEE_J_IT, 
  title={Optimal soft decision block decoders based on fast {Hadamard} transform}, 
  year={1986},
  volume={32},
  number={3},
  pages={355-364},
}

@ARTICLE{Sidelnikov,
  author={V. M. {Sidel’nikov} and A. S. {Pershakov}},
  journal={ Problemy peredachi informatsii}, 
  title={Decoding of {Reed-Muller} codes with a large number of errors}, 
  year={1992},
  volume={28},
  number={3},
  pages={80-94},
}

@ARTICLE{DumerList,
  author={Dumer, I. and Shabunov, K.},
  journal=IEEE_J_IT, 
  title={Soft-decision decoding of {Reed-Muller} codes: recursive lists}, 
  year={2006},
  volume={52},
  number={3},
  pages={1260-1266},
  doi={10.1109/TIT.2005.864443}}

@ARTICLE{RPA,
  author={Ye, Min and Abbe, Emmanuel},
  journal=IEEE_J_IT, 
  title={Recursive Projection-Aggregation Decoding of {Reed-Muller} Codes}, 
  year={2020},
  volume={66},
  number={8},
  pages={4948-4965},
  doi={10.1109/TIT.2020.2977917}}

@INPROCEEDINGS{RPA_BP,
  author={Lian, Mengke and Häger, Christian and Pfister, Henry D.},
  booktitle={IEEE International Symposium on Information Theory (ISIT)}, 
  title={Decoding {Reed–Muller} Codes Using Redundant Code Constraints}, 
  year={2020},
  volume={},
  number={},
  pages={42-47},
  doi={10.1109/ISIT44484.2020.9174087}}

@article{Pruning,
  author    = {Qin Huang and
               Bin Zhang},
  title     = {Pruned Collapsed Projection-Aggregation Decoding of {Reed-Muller} Codes},
  journal   = {CoRR},
  volume    = {abs/2105.11878},
  year      = {2021},
  url       = {https://arxiv.org/abs/2105.11878},
  eprinttype = {arXiv},
  eprint    = {2105.11878},
  bibsource = {dblp computer science bibliography, https://dblp.org}
}

@ARTICLE{RM_S,
  author={Abbe, Emmanuel and Shpilka, Amir and Ye, Min},
  journal=IEEE_J_IT, 
  title={{Reed–Muller} Codes: Theory and Algorithms}, 
  year={2021},
  volume={67},
  number={6},
  pages={3251-3277},
  doi={10.1109/TIT.2020.3004749}}

@book{richardson2008modern,
  title={Modern coding theory},
  author={Richardson, Tom and Urbanke, Ruediger},
  year={2008},
  publisher={Cambridge university press}
}

@ARTICLE{richardson2001deldpc,
  author={Richardson, T.J. and Urbanke, R.L.},
  journal=IEEE_J_IT, 
  title={The capacity of low-density parity-check codes under message-passing decoding}, 
  year={2001},
  volume={47},
  number={2},
  pages={599-618},
  doi={10.1109/18.910577}}

@ARTICLE{saeyoung2001Gaussianapproximation,
  author={Sae-Young Chung and Richardson, T.J. and Urbanke, R.L.},
  journal=IEEE_J_IT, 
  title={Analysis of sum-product decoding of low-density parity-check codes using a {Gaussian} approximation}, 
  year={2001},
  volume={47},
  number={2},
  pages={657-670},
  doi={10.1109/18.910580}}

@ARTICLE{dumer2004recursive,
  author={Dumer, I.},
  journal=IEEE_J_IT, 
  title={Recursive decoding and its performance for low-rate {Reed-Muller} codes}, 
  year={2004},
  volume={50},
  number={5},
  pages={811-823},
  doi={10.1109/TIT.2004.826632}}

@INPROCEEDINGS{rameshwar2024upper,
  author={Rameshwar, V. Arvind and Lalitha, V.},
  booktitle={IEEE International Symposium on Information Theory (ISIT)}, 
  title={An Upper Bound on the Error Probability of {RPA} Decoding of {Reed-Muller} Codes Over the {BSC}}, 
  year={2025},
  volume={},
  number={},
  pages={1-6},
  doi={10.1109/ISIT63088.2025.11195324}}

@ARTICLE{zhang2025errorpatternpa,
  author={Zhang, Bin and Chen, Fanyun and Huang, Qin},
  journal=IEEE_J_IT, 
  title={Coset Error Pattern in Projection-Aggregation Decoding}, 
  year={2025},
  volume={},
  number={},
  pages={1-1},
  doi={10.1109/TIT.2025.3576779}}

@article{pfister2025capacity,
  title={Capacity on {BMS} Channels via Code Symmetry and Nesting},
  author={Pfister, Henry D and Reeves, Galen},
  journal={arXiv preprint arXiv:2504.15394},
  year={2025}
}

@INPROCEEDINGS{li2024layeredCPA,
  author={Li, Jiajie and Gross, Warren J.},
  booktitle={58th Asilomar Conference on Signals, Systems, and Computers}, 
  title={A Layered {CPA} Decoder for {Reed-Muller} Codes}, 
  year={2024},
  volume={},
  number={},
  pages={985-989},
  doi={10.1109/IEEECONF60004.2024.10942623}}

@ARTICLE{Li2023improvePAlist,
  author={Li, Jiajie and Zhou, Huayi and Jalaleddine, Marwan and Gross, Warren J.},
  journal=IEEE_J_COM, 
  title={Reduced-Complexity Projection-Aggregation List Decoder for {Reed-Muller} Codes}, 
  year={2025},
  volume={73},
  number={3},
  pages={1458-1473},
  doi={10.1109/TCOMM.2024.3459851}}

@article{abbe2024polynomial,
  title={Polynomial {Freiman-Ruzsa}, {Reed-Muller} codes and {Shannon} capacity},
  author={Abbe, Emmanuel and Sandon, Colin and Shashkov, Vladyslav and Viazovska, Maryna},
  journal={arXiv preprint arXiv:2411.13493},
  year={2024}
}

@ARTICLE{hagenauer2002iterative,
  author={Hagenauer, J. and Offer, E. and Papke, L.},
  journal=IEEE_J_IT, 
  title={Iterative decoding of binary block and convolutional codes}, 
  year={1996},
  volume={42},
  number={2},
  pages={429-445},
  doi={10.1109/18.485714}}

@inproceedings{luby1998analysis,
  title={Analysis of low density codes and improved designs using irregular graphs},
  author={Luby, Michael and Mitzenmacher, Michael and Shokrollah, A and Spielman, Daniel},
  booktitle={Proceedings of the thirtieth annual ACM symposium on Theory of computing},
  pages={249--258},
  year={1998}
}

@ARTICLE{forney2002codes,
  author={Forney, G.D.},
  journal=IEEE_J_IT, 
  title={Codes on graphs: normal realizations}, 
  year={2001},
  volume={47},
  number={2},
  pages={520-548},
  doi={10.1109/18.910573}}

@ARTICLE{hartmann1976optimum,
  author={Hartmann, C. and Rudolph, L.},
  journal=IEEE_J_IT, 
  title={An optimum symbol-by-symbol decoding rule for linear codes}, 
  year={1976},
  volume={22},
  number={5},
  pages={514-517},
  doi={10.1109/TIT.1976.1055617}}

@ARTICLE{battail1979replication,
  author={Battail, G. and Decouvelaere, M. and Godlewski, P.},
  journal=IEEE_J_IT, 
  title={Replication decoding}, 
  year={1979},
  volume={25},
  number={3},
  pages={332-345},
  doi={10.1109/TIT.1979.1056035}}

@electronic{upper_bound_normal_distribution,
  title         = {Upper and lower bounds for the normal distribution function},
  url           = {https://www.johndcook.com/blog/norm-dist-bounds/},
  note          = {accessed on Dec. 28, 2025}
}

@article{fathollahi2026error,
  title={On the Error Probability of {RPA} Decoding of {Reed-Muller} Codes over {BMS} Channels},
  author={Fathollahi, Dorsa and Rameshwar, V Arvind and Lalitha, V},
  journal={arXiv preprint arXiv:2601.09581},
  year={2026}
}

@electronic{sum_dep_normal_may_not_be_normal,
  title         = {Sums of normal random variables need not be normal},
  url           = {https://planetmath.org/sumsofnormalrandomvariablesneednotbenormal},
  year          = {2013},
  note           = {accessed on Jan. 29, 2026}
}

@inproceedings{saptharishi2016efficiently,
  title={Efficiently decoding {Reed-Muller} codes from random errors},
  author={Saptharishi, Ramprasad and Shpilka, Amir and Volk, Ben Lee},
  booktitle={Proceedings of the forty-eighth annual ACM symposium on Theory of Computing},
  pages={227--235},
  year={2016}
}

@ARTICLE{saptharishi2017efficiently,
  author={Saptharishi, Ramprasad and Shpilka, Amir and Volk, Ben Lee},
  journal=IEEE_J_IT, 
  title={Efficiently Decoding {Reed–Muller} Codes From Random Errors}, 
  year={2017},
  volume={63},
  number={4},
  pages={1954-1960},
  doi={10.1109/TIT.2017.2671410}}

@INPROCEEDINGS{zhang2025minimum,
  author={Zhang, Bin and Huang, Qin},
  booktitle={IEEE Information Theory Workshop (ITW)}, 
  title={Minimum Distance Decoding for {Reed-Muller} Codes using Projection-Aggregation}, 
  year={2025},
  volume={},
  number={},
  pages={680-685},
  doi={10.1109/ITW62417.2025.11240441}}

@ARTICLE{Op_PCPA,
  author={Li, Jiajie and Gross, Warren J.},
  journal=IEEE_J_COML, 
  title={Optimization and Simplification of {PCPA} Decoder for {Reed-Muller} Codes}, 
  year={2022},
  volume={26},
  number={6},
  pages={1206-1210},
  doi={10.1109/LCOMM.2022.3163622}}

@book{roy2021series,
  title={Series and Products in the Development of Mathematics: Volume 1},
  author={Roy, Ranjan},
  year={2021},
  publisher={Cambridge University Press}
}
